\documentclass[conference,a4paper]{IEEEtran}
\usepackage{mleftright}       
\mleftright                   
\usepackage[utf8]{inputenc}
\usepackage[T1]{fontenc}
\usepackage[cmex10]{amsmath}
\usepackage{amssymb}
\usepackage{dsfont}
\usepackage{xcolor}
\usepackage{tabu}
\usepackage{multirow}
\usepackage{amsthm}
\usepackage{float}
\usepackage{verbatim}
\usepackage[maxbibnames=99,style=ieee,sorting=nyt,citestyle=numeric-comp]{biblatex}

\bibliography{bibliography.bib}

\title{On the Information Capacity of Nearest Neighbor Representations}
\author{%
   \IEEEauthorblockN{\textbf{Kordag Mehmet Kilic}\IEEEauthorrefmark{1}, \textbf{Jin Sima}\IEEEauthorrefmark{2} and \textbf{Jehoshua Bruck}\IEEEauthorrefmark{1}}
   \IEEEauthorblockA{\IEEEauthorrefmark{1}%
   Electrical Engineering, California Institute of Technology, USA, \texttt{\{kkilic,bruck\}@caltech.edu}
   }
   \IEEEauthorblockA{\IEEEauthorrefmark{2}%
   Electrical and Computer Engineering, University of Illinois Urbana-Champaign, USA, \texttt{jsima@illinois.edu}
   }
 }

\newtheorem{theorem}{Theorem}
\newtheorem{corollary}{Corollary}[theorem]
\newtheorem{lemma}{Lemma}
\newtheorem{proposition}{Proposition}
\newtheorem{definition}{Definition}

\DeclareMathOperator {\diag}{diag}

\usepackage{tikz}[border=2mm]
\usepackage{float}
\usetikzlibrary{shapes.geometric,positioning,shapes.gates.logic.US,tikzmark}
\AtBeginBibliography{\small}
\tikzset{
    ->,
    gate/.style={draw=black,fill=#1,minimum width=6mm,circle},
    square/.style={regular polygon,regular polygon sides=4},
    every pin edge/.style={draw=black},
    GateCfg/.style={
            logic gate inputs={normal,normal,normal},
            draw,
            scale=2
        }
}

\begin{document}

\maketitle
\thispagestyle{plain}
\pagestyle{plain}

\begin{abstract}
    The \textit{von Neumann Computer Architecture} has a distinction between computation and memory. In contrast, the brain has an integrated architecture where computation and memory are indistinguishable. Motivated by the architecture of the brain, we propose a model of \textit{associative computation} where memory is defined by a set of vectors in $\mathbb{R}^n$ (that we call \textit{anchors}), computation is performed by convergence from an input vector to a nearest neighbor anchor, and the output is a label associated with an anchor. Specifically, in this paper, we study the representation of Boolean functions in the associative computation model, where the inputs are binary vectors and the corresponding outputs are the labels ($0$ or $1$) of the nearest neighbor anchors. The information capacity of a Boolean function in this model is associated with two quantities: \textit{(i)} the number of anchors (called \textit{Nearest Neighbor (NN) Complexity}) and \textit{(ii)} the maximal number of bits representing entries of anchors (called \textit{Resolution}). We study symmetric Boolean functions and present constructions that have optimal NN complexity and resolution.
\end{abstract}
\section{Introduction}
Inspired by the architecture of the brain, we propose a model of \textit{associative computation} where memory is defined by a set a of vectors in $\mathbb{R}^n$ (that we call \textit{anchors}), computation is performed by convergence from an input vector to a Nearest Neighbor anchor, and the output is a label associated with an anchor. This paradigm relates to the ability of the brain to quickly classify objects and map those to the syntax of natural languages, for example, it is instantaneous and effortless, even for a child, to recognize a `dog', a `cat', or a `car'.  
 
In this paper, we study the representation of \textit{Symmetric Boolean Functions} in the associative computation model, where the inputs are binary vectors and the corresponding outputs are the labels ($0$ or $1$) of the Nearest Neighbor anchors. The information capacity of a Boolean function in this model is associated with two quantities: \textit{(i)} the number of anchors (called \textit{Nearest Neighbor (NN) Complexity} \cite{hajnal2022nearest}) and \textit{(ii)} the maximal number of bits representing entries of anchors (called \textit{Resolution}). Metaphorically speaking, we are interested in the number of ‘words’ (anchors) and the size of the ‘alphabet’ (resolution) required to represent a Boolean function. 

Let $d(a,b)$ denote the Euclidean distance between the vectors $a,b \in \mathbb{R}^n$. The Nearest Neighbor Representation and Complexity are defined in the following manner.

\begin{definition}
    The \textup{Nearest Neighbor (NN) Representation} of a Boolean function $f$ is a set of anchors consisting of the disjoint subsets $(P,N)$ of $\mathbb{R}^n$ such that for every $X \in \{0,1\}^n$ with $f(X) = 1$, there exists $p \in P$ such that for every $n \in N$, $d(X,p) < d(X,n)$, and vice versa. The \textup{size} of the NN representation is $|P \cup N|$.
\end{definition}

Namely, the size of the NN representation is the number of anchors. Naturally, we are interested in representations of minimal size.

\begin{definition}
    The \textup{Nearest Neighbor Complexity} of a Boolean function $f$ is the minimum size over all NN representations of $f$, denoted by $NN(f)$.
\end{definition}

Some $2$-input examples that can be easily verified are given below. We say $f(X) = 1$ (or $0$) and their corresponding anchors are \textit{positive} (or \textit{negative}).

\begin{figure}[h]
    \centering
    \begin{tabular}{c c|c|c|c}
        $X_1$ & $X_2$ & $\text{AND}(X)$ & $\text{OR}(X)$ & $\text{XOR}(X)$ \\
        \hline
         0 & 0 & \textcolor{red}{\textbf{0}} & \textcolor{red}{\textbf{0}} & \textcolor{red}{\textbf{0}} \\
         0 & 1 & \textcolor{red}{\textbf{0}} & \textcolor{blue}{1} & \textcolor{blue}{1}  \\
         1 & 0 & \textcolor{red}{\textbf{0}} & \textcolor{blue}{1}  & \textcolor{blue}{1}  \\ 
         1 & 1 & \textcolor{blue}{1} & \textcolor{blue}{1} & \textcolor{red}{\textbf{0}} \\
         \hline
         \multicolumn{2}{c|}{Anchors} & $\left(
    \begin{matrix}
        \textcolor{red}{\textbf{0.5}}  & \textcolor{red}{\textbf{0.5}}   \\
        \textcolor{blue}{1} & \textcolor{blue}{1}
    \end{matrix}
    \right)$ & $\left(
    \begin{matrix}
        \textcolor{red}{\textbf{0}}   & \textcolor{red}{\textbf{0}}   \\
        \textcolor{blue}{0.5} & \textcolor{blue}{0.5}
    \end{matrix}
    \right)$ & $\left(
    \begin{matrix}
        \textcolor{red}{\textbf{0}}   & \textcolor{red}{\textbf{0}}   \\
        \textcolor{blue}{0.5} & \textcolor{blue}{0.5} \\
        \textcolor{red}{\textbf{1}} & \textcolor{red}{\textbf{1}}
    \end{matrix}
    \right)$
    \end{tabular}
    \caption{The $2$-input Boolean $\text{AND}(X) = x_1 \land x_2$, $\text{OR}(X) = x_1 \lor x_2$, and $\text{XOR}(X) = x_1 \oplus x_2$ functions, their corresponding truth values, and their NN representations. The positive values and anchors are in \textcolor{blue}{blue} while the negative values and anchors are in \textcolor{red}{\textbf{red}}.}
    \label{fig:and_or_xor}
\end{figure}

The examples given in Fig. \ref{fig:and_or_xor} belong to a class of Boolean functions called \textit{Symmetric Boolean Functions}. Let $|X|$ denote the number of $1$s in a binary vector $X$.

\begin{definition}
    A Boolean function is called \textup{symmetric} $f(X) = f(\sigma(X))$ where $\sigma(.)$ can be any permutation function.  Namely, a symmetric Boolean function $f(X)$ is a function of $|X|$. 
\end{definition}

We emphasize that although Boolean functions are evaluated over binary vectors, the anchor points can have real valued entries. Nevertheless, this need not to be the case as we can consider all the $2^n$ nodes of the Boolean hypercube as anchors and assign them to $P$ or $N$ based on the values of $f$. This implies that $NN(f) \leq 2^n$ for any $n$-input Boolean function. Specifically, we define the \textit{Boolean Nearest Neighbor Complexity (BNN)} for the case where we restrict all anchors to be binary vectors. Thus, it is easy to see that $NN(f) \leq BNN(f) \leq 2^n$ \cite{hajnal2022nearest}. This upper bound on the NN complexity is essentially loose for almost all Boolean functions.

\begin{figure}[h]
    \centering
    \begin{tabular}{c c|c|c|c}
        $X_1$ & $X_2$ & $|X|$ & $\text{XOR}(X)$ & BNN Representation \\
        \hline
         0 & 0 & 0 & \textcolor{red}{\textbf{0}} & \multirow{4}{*}{$\left(
    \begin{matrix}
        \textcolor{red}{\textbf{0}}  & \textcolor{red}{\textbf{0}}   \\
        \textcolor{blue}{0} & \textcolor{blue}{1} \\
        \textcolor{blue}{1}  & \textcolor{blue}{0}   \\
        \textcolor{red}{\textbf{1}} & \textcolor{red}{\textbf{1}}
    \end{matrix}
    \right)$} \\
         0 & 1 & 1 & \textcolor{blue}{1} &  \\
         1 & 0 & 1 & \textcolor{blue}{1}  &  \\ 
         1 & 1 & 2 & \textcolor{red}{\textbf{0}} & \\
    \end{tabular}
    \caption{The $2$-input XOR functions with its optimal BNN representation. It is known that $BNN(\text{XOR}) = 4$ whereas $NN(\text{XOR}) = 3$.}
    \label{fig:xor_bnn_nn}
\end{figure}

The discrepancy between NN and BNN is illustrated for the $\text{XOR}$ function in Fig. \ref{fig:xor_bnn_nn}. The generalization of $\text{XOR}$ to $n$ variables is called the \text{PARITY} function. $\text{PARITY}(X)$ is 1 if the number of 1s in $X$ is odd, i.e., $\text{PARITY}(X) = x_1 \oplus x_2 \oplus \cdots \oplus x_n$ and we have a dramatic gap between NN and BNN.

\begin{theorem}[\cite{hajnal2022nearest}]
\label{th:parity_bounds}
Let $f$ be the \textup{PARITY} of $n$-inputs. Then,
\begin{itemize}
    \item $NN(f) = n + 1$
    \item $BNN(f) = 2^n$
\end{itemize}
\end{theorem}

The fundamental difference between NN and BNN is the number of bits that are used to represent the anchor entries. To quantify it, we define the $\textit{resolution}$ of an NN representation. Loosely speaking, the number of bits to represent a number is its resolution and the resolution of a representation is the maximum resolution for a single entry for all anchor points. We say that $A \in \mathbb{R}^{m\times n}$ is an anchor matrix if each row corresponds to an anchor $a_i$ for an NN representation of an $n$-input Boolean function with $m$ anchors. We remark that we can assume that $A \in \mathbb{Q}^{m \times n}$ without loss of generality.

\begin{definition}
    The \textup{resolution ($RES$) of a rational number} $a/b$ is $RES(a/b) = \lceil \max\{\log_2{|a+1|},\log_2{|b+1|}\}\rceil$ where $a,b \in \mathbb{Z}$, $b \neq 0$, and they are coprime.
    
    For a matrix $A \in \mathbb{Q}^{m\times n}$, $RES(A) = \max_{i,j} RES(a_{ij})$. The \textup{resolution of an NN representation} is $RES(A)$ where $A$ is the corresponding anchor matrix.
\end{definition}

Given the PARITY example in Theorem \ref{th:parity_bounds}, we observe that there is a trade-off between the NN complexity and the associated resolution. The $BNN(\text{PARITY})$ is exponential in $n$ with resolution $1$ while $NN(\text{PARITY})$ is linear in $n$ with resolution $\lceil\log_2{(n+1)}\rceil$. 

Our paper is inspired by the work of Hajnal, Liu and Tur{\'a}n \cite{hajnal2022nearest}. Their paper provided several NN complexity results, including the construction for PARITY (see Theorem \ref{th:parity_bounds}), however, it only proposed the suboptimal PARITY-based construction for general symmetric functions (see Proposition \ref{prop:sym_bound}). We briefly mention some additional relevant background: Nearest neighbors classification is a well-studied topic in Information Theory and Machine Learning \cite{cover1967nearest,hart1968condensed}. Optimizing the representation of a set by using NN was first discussed in the context of minimizing the size of a training set, namely, finding a minimal training set that NN represents the original training set \cite{salzberg1995best,wilfong1991nearest}. The idea to represent Boolean functions via the NN paradigm was first studied in \cite{globig1996case} where BNN was considered with distance measures that are chosen to optimize the complexity of the NN representation. This work was extended to inclusion-based similarity \cite{satoh1998analysis} and it was proved that it provides polynomial size representations for DNF and CNF formulas.

In this paper we focus on the study of the NN representations of symmetric Boolean functions. Specifically, we construct optimal size NN representations for symmetric Boolean functions. These functions are useful to prove many complexity results about Boolean functions \cite{bruck1990harmonic,haastad1995top,stockmeyer1976combinational} and more importantly, in the context of nearest neighbors, any $n$-input Boolean function can be interpreted as an $2^n$-input symmetric Boolean function \cite{siu1991depth}. Therefore, results on symmetric Boolean functions can provide insights about our understanding of Boolean function complexity. Additionally, we study the trade-offs between the NN complexity (the number of anchors in the NN representation) and the NN resolution (the maximal number of bits to represent entries of the anchors). A non-intuitive discovery is that for symmetric Boolean functions, optimality in complexity and resolution is achieved by anchors that are asymmetric in their entries!

\subsection{The NN Complexity of Symmetric Functions}

The PARITY construction mentioned in Theorem \ref{th:parity_bounds} can be extended to arbitrary symmetric Boolean functions. The idea is to assign an anchor point to each value of $|X|$. We have the following construction of an NN representation with an anchor matrix $A$ for the given symmetric Boolean function $f(X)$.

\begin{equation}
    \label{eq:example_construction}
    \begin{tabular}{c|c c}
        $|X|$ & $f(X)$ \\
        \cline{1-2}
         0 & \textcolor{red}{\textbf{0}} & $\rightarrow a_1$ \\
         1 & \textcolor{red}{\textbf{0}} & $\rightarrow a_2$\\
         2 & \textcolor{blue}{1} & $\rightarrow a_3$ \\ 
         3 & \textcolor{blue}{1}  & $\rightarrow a_4$ \\
         4 & \textcolor{blue}{1}  & $\rightarrow a_5$ \\
         5 & \textcolor{red}{\textbf{0}} & $\rightarrow a_6$ 
    \end{tabular}\hphantom{aa} A =
    \begin{bmatrix}
        \textcolor{red}{\textbf{0}}   & \textcolor{red}{\textbf{0}}  & \textcolor{red}{\textbf{0}}   & \textcolor{red}{\textbf{0}}   & \textcolor{red}{\textbf{0}}   \\
        \textcolor{red}{\textbf{0.2}} & \textcolor{red}{\textbf{0.2}} & \textcolor{red}{\textbf{0.2}} & \textcolor{red}{\textbf{0.2}} & \textcolor{red}{\textbf{0.2}} \\
        \textcolor{blue}{0.4} & \textcolor{blue}{0.4} & \textcolor{blue}{0.4} & \textcolor{blue}{0.4} & \textcolor{blue}{0.4} \\
        \textcolor{blue}{0.6} & \textcolor{blue}{0.6} & \textcolor{blue}{0.6} & \textcolor{blue}{0.6} & \textcolor{blue}{0.6} \\
        \textcolor{blue}{0.8} & \textcolor{blue}{0.8} & \textcolor{blue}{0.8} & \textcolor{blue}{0.8} & \textcolor{blue}{0.8} \\
        \textcolor{red}{\textbf{1}}  &  \textcolor{red}{\textbf{1}} &   \textcolor{red}{\textbf{1}} &   \textcolor{red}{\textbf{1}} &   \textcolor{red}{\textbf{1}}
    \end{bmatrix}
\end{equation}

\begin{proposition}[\cite{hajnal2022nearest}]
\label{prop:sym_bound}
Let $f$ be a symmetric Boolean function with $n$ input. Then, $NN(f) \leq n + 1$.
\end{proposition}

Proposition \ref{prop:sym_bound} is proven by construction and in general, the matrix $A \in \mathbb{R}^{(n+1) \times n}$ contains the anchor $a_i = ((i-1)/n,\dots,(i-1)/n)$ where $P$ and $N$ is a partition of $\{a_1,\dots,a_{n+1}\}$. Regardless of the function itself, this construction requires $n+1$ anchors and we call it the \textit{\textup{PARITY}-based construction}. In contrast, we know that there are examples of symmetric Boolean functions that require a smaller number of anchors and Proposition \ref{prop:sym_bound} is evidently far from being optimal.

\begin{figure}[h]
    \centering
    \begin{tabular}{c|c|c}
        $|X|$ & $\text{AND}(X)$ & $\text{OR}(X)$ \\
        \hline
         0  & \textcolor{red}{\textbf{0}} & \textcolor{red}{\textbf{0}}  \\
         1 & \textcolor{red}{\textbf{0}} & \textcolor{blue}{1}  \\
         2  & \textcolor{blue}{1} & \textcolor{blue}{1} \\
         \hline
         Anchors & $\left(
    \begin{matrix}
        \textcolor{red}{\textbf{0}}  & \textcolor{red}{\textbf{0}}   \\
        \textcolor{red}{\textbf{0.5}} & \textcolor{red}{\textbf{0.5}} \\
        \textcolor{blue}{1}  & \textcolor{blue}{1}
    \end{matrix}
    \right)$ & $\left(
    \begin{matrix}
        \textcolor{red}{\textbf{0}}  & \textcolor{red}{\textbf{0}}   \\
        \textcolor{blue}{0.5} & \textcolor{blue}{0.5} \\
        \textcolor{blue}{1}  & \textcolor{blue}{1}
    \end{matrix}
    \right)$
    \end{tabular}
    \caption{The PARITY-based constructions for $2$-input AND and OR functions. Clearly, the construction given in Fig. \ref{fig:and_or_xor} has smaller size and the same resolution.}
    \label{fig:and_or_parity_based}
\end{figure}

One can notice that there is a single transition from $0$s to $1$s in $f(X)$ for both $\text{AND}$ and $\text{OR}$ when they are enumerated by $|X|$. Furthermore, the correspondence between $f(X)$ and $|X|$ is particularly useful to measure the complexity of symmetric Boolean functions in Circuit Complexity Theory \cite{muroga1959logical,minnick1961linear,siu1991power}. These observations motivate us to define the notion of an \textit{interval} for a symmetric Boolean function in order to measure the NN complexity.

\begin{definition}
Let $a \leq b \leq n$ be some non-negative integers. An \textup{interval $[a,b]$ for an $n$-input symmetric Boolean function} is defined as follows: 
\begin{enumerate}
    \item $f(X)$ is constant for $|X| \in [a,b]$.
    \item If $a \neq 0$ and $b \neq n$, $f(X_1) \neq f(X_2)$ for any $|X_1| \in [a,b]$ and
    \begin{itemize}
        \item If $a > 0$ and $b < n$, $|X_2| = a-1$ and $|X_2| = b+1$.
        \item If $a = 0$ and $b < n$, $|X_2| = b+1$.
        \item If $a > 0$ and $b = n$, $|X_2| = a-1$.
    \end{itemize}
\end{enumerate}
 The quantity $I(f)$ is the total number of intervals for a symmetric Boolean function $f$.
\end{definition}

We usually refer to an interval $[a,b]$ for an $n$-input symmetric Boolean function shortly as an \textup{interval}. The example function in Eq. \ref{eq:example_construction} and $XOR$ have $I(f)=3$ (see Fig.\ref{fig:xor_bnn_nn}) while $AND$ and $OR$ have $I(f) = 2$ (see Fig. \ref{fig:and_or_parity_based}).

It seems natural to simplify the construction in Proposition \ref{prop:sym_bound} by assigning an anchor to each interval to reduce the size of the representation to $I(f)$. We call this an \textbf{\textit{$I(f)$ interval-anchor assignment}} where there is a one-to-one map between an anchor and an interval of $f(X)$. We assign the $i^{th}$ interval $[a,b]$ to the anchor $a_i$ and enumerate the beginning and the end of the interval assigned to it by $I_{i-1} + 1$ and $I_{i}$ respectively by using the values of $|X|$. We take $I_0 = -1$ for consistency. The function below is an example where $f(X)$ is a 6-input symmetric function and $I(f) = 3$ and $(a_1,a_2,a_3)$ is an $I(f)$ interval-anchor assignment.
\begin{equation}
    \begin{tabular}{c|c c c}
        $|X|$ & $f(X)$ \\
        \cline{1-2}
         0 & \textcolor{red}{\textbf{0}} & $\rightarrow a_1$ & $I_0 + 1 = 0$\\
         1 & \textcolor{red}{\textbf{0}}  & $\rightarrow a_1$& $I_1 = 1$ \\
         2 & \textcolor{blue}{1} & $\rightarrow a_2$ & $I_1 + 1 = 2$\\ 
         3 & \textcolor{blue}{1} & $\rightarrow a_2$ \\
         4 & \textcolor{blue}{1} & $\rightarrow a_2$ & $I_2 = 4$\\
         5 & \textcolor{red}{\textbf{0}}  & $\rightarrow a_3$ & $I_2 + 1 = I_3 = 5$
    \end{tabular}
\end{equation}

It might be possible to extend the construction given in Proposition \ref{prop:sym_bound} and find an $I(f)$ interval-anchor assignment by computing and perturbing the centroids of anchors belonging to each interval. We call this idea \textit{\textup{PARITY}-based extension}. The anchors in PARITY-based extensions are symmetric, i.e., they are equal. Consider the construction given in Eq. \eqref{eq:example_construction}. One can obtain a 3-anchor NN representation by the matrix $A' \in \mathbb{R}^{3 \times 5}$ using the matrix $A$.
\begin{equation}
    \label{eq:parity_ext_example}
    \begin{tabular}{c|c c}
        $|X|$ & $f(X)$ \\
        \cline{1-2}
         0 & \textcolor{red}{\textbf{0}} & $\rightarrow a'_1$ \\
         1 & \textcolor{red}{\textbf{0}} & \hphantom{a}\\
         2 & \textcolor{blue}{1} & \hphantom{a}\\
         3 & \textcolor{blue}{1} & $\rightarrow a'_2$ \\
         4 & \textcolor{blue}{1} & \hphantom{a} \\
         5 & \textcolor{red}{\textbf{0}} & $\rightarrow a'_3$  
    \end{tabular}\hphantom{aa} A' =
    \begin{bmatrix}
        \textcolor{red}{\textbf{0.1}} & \textcolor{red}{\textbf{0.1}}  & \textcolor{red}{\textbf{0.1}}  & \textcolor{red}{\textbf{0.1}}  & \textcolor{red}{\textbf{0.1}}  \\
        \textcolor{blue}{0.6} & \textcolor{blue}{0.6} & \textcolor{blue}{0.6} & \textcolor{blue}{0.6} & \textcolor{blue}{0.6} \\
        \textcolor{red}{\textbf{1.1}}  & \textcolor{red}{\textbf{1.1}}  & \textcolor{red}{\textbf{1.1}}  & \textcolor{red}{\textbf{1.1}} & \textcolor{red}{\textbf{1.1}} 
    \end{bmatrix}
\end{equation}

Unfortunately, PARITY-based extensions cannot be applied to all symmetric Boolean functions. We give an example where there is no PARITY-based extension in Appendix \ref{apx:counterexample}. Informally, symmetric Boolean functions treat each input $x_i$ with the same importance and we show that by breaking up the symmetry in the anchor entries and taking the whole NN representation farther away from the Boolean hypercube, we can reduce the size of representation to $I(f)$. 

Is $NN(f) = I(f)$ for symmetric Boolean functions? To conclude that we need a corresponding lower bound for an arbitrary symmetric Boolean function. It is still an open problem, however, we are able to prove it for \textit{periodic} symmetric Boolean functions.

\begin{definition}
A symmetric Boolean function is called \textup{periodic} if each interval has the same length, which is denoted by $T$. $T$ is also called \textup{period}.
\end{definition}

More precisely, if $f$ is a periodic symmetric Boolean function with period $T$, then for $(k-1)T \leq |X| < kT$, $f(X) = 0$ (or $1$) for odd $k$ and $f(X) = 1$ (or $0$) for even $k$. $\text{PARITY}$ is a periodic symmetric function with period $T = 1$. Periodic symmetric Boolean functions are useful in other works as well (see the Complete Quadratic function in \cite{bruck1990harmonic} where $T = 2$).

\subsection{NN Representations of Linear Threshold Functions}
\label{sec:lt}

A linear threshold function is a weighted summation of binary inputs fed into a threshold operation, i.e., $\mathds{1}\{\sum_{i=1}^n w_ix_i \geq b\}$ where $\mathds{1}\{.\}$ is an indicator function that evaluates $1$ if the inside expression is true and $0$ otherwise. $b$ is called the \textit{threshold} and without loss of generality, $w_i \in \mathbb{Z}$. For non-constant linear threshold functions, an optimal hyperplane argument using support vectors implies that $NN(f) = 2$ and the converse also holds in that any Boolean function with $NN(f) = 2$ is a linear threshold function. These functions are interesting because they are the building blocks of neural networks. 

 We define an $n$-input \textit{symmetric linear threshold function} as a symmetric Boolean function with $I(f) = 2$ intervals, i.e., $f(X) = \mathds{1}\{|X| \geq b\}$, namely, $w_i = 1$ for all $i$. Interestingly, depending on the value of $b$, these functions can have different orders of BNN complexities.

\begin{theorem}[\cite{hajnal2022nearest}]
    \label{th:bnn_th}
    Let $f$ be an $n$-input symmetric linear threshold function with threshold $b$. Then,
    \begin{itemize}
        \item If $b = n/2$ and $n$ is odd, $BNN(f) = 2$.
        \item If $b = n/2$ and $n$ is even, $BNN(f) \leq n/2 + 2$.
        \item If $b = \lfloor n/3 \rfloor$, $BNN(f) = 2^{\Omega(n)}$.
    \end{itemize}
\end{theorem}

If all the anchors are on the Boolean hypercube, then the NN complexity can be as large as $2^{\Omega(n)}$. Nevertheless, in this work, we show that the resolution of a $2$-anchor NN representation of a linear threshold function is bounded by the resolution of the weights. This implies that there are constant resolution NN representations for all symmetric linear threshold functions. In general, for linear threshold functions, the resolution is $O(n\log{n})$ \cite{alon1997anti,haastad1994size,muroga1971threshold}.

\subsection{Contributions and Organizations}

Our key contributions are:

\begin{itemize}
    \item New constructions are presented for arbitrary $n$-input symmetric Boolean functions with $I(f)$ anchors and $O(\log{n})$ resolution.
    \item The new constructions have optimal NN complexity for all periodic symmetric Boolean functions.
    \item There is always an optimal NN representation both in size and resolution for any $n$-input symmetric linear threshold function.
    \item The NN representations of some symmetric Boolean functions require the resolution to be $\Omega(\log{n})$ so that the new constructions is optimal in resolution.
\end{itemize}

The organization of the paper is as follows. In Section \ref{sec:upper_bound}, we present an explicit construction and thus, an upper bound for the NN complexity of any symmetric Boolean function. In Section \ref{sec:lower_bound}, we prove a lower bound for periodic symmetric Boolean functions and conclude that $NN(f) = I(f)$ for this class of functions. Later, in Section \ref{sec:resolution}, we study the resolution of our NN constructions, prove that it is $O(\log{n})$ and that it is optimal for some symmetric Boolean functions. In addition, we present optimal NN representation both in size and resolution for symmetric linear threshold functions. The appendices include the details of the proofs.

\section{The Upper Bound on the NN Complexity of Symmetric Boolean Functions}
\label{sec:upper_bound}

In this section, we present an explicit construction for the NN representation of symmetric Boolean functions with $I(f)$ intervals.

\begin{theorem}
\label{th:upper}
For an $n$-input symmetric Boolean function $f$, $NN(f) \leq I(f)$. 

Moreover, a matrix $B = \mathds{1} + \epsilon M$ where $B \in \mathbb{R}^{(I(f)-1) \times n}$ can always be used for a construction given that $\epsilon > 0$ is sufficiently small where $\mathds{1}$ is an all-one matrix and $M \in \mathbb{R}^{(I(f)-1)\times n}$ is a full row rank matrix.
\end{theorem}

For arbitrary symmetric Boolean functions, we first derive necessary and sufficient conditions for an $I(f)$ interval-anchor assignment.

\begin{lemma}
\label{lem:ns_cond}

Let $f$ be an $n$-input symmetric Boolean function and $A \in \mathbb{R}^{I(f)\times n}$ be an anchor matrix for an $I(f)$ interval-anchor assignment. Then, the following condition is necessary and sufficient for any $i \in \{2,\dots,I(f)\}$ and $k < i$:
\begin{align}
    \label{eq:ns_condition}
    \max_{X \in \{0,1\}^n : |X| = I_{k}} &\sum_{j=1}^n (a_{ij} - a_{kj})x_j  < \frac{1}{2}\sum_{j=1}^n (a_{ij}^2 - a_{kj}^2) \nonumber \\
     &< \min_{X \in \{0,1\}^n : |X| = I_{i-1}+1} \sum_{j=1}^n (a_{ij} -a_{kj})x_j 
\end{align}

\end{lemma}

In general, all $a_{ij}$s can be different compared to $\text{PARITY}$-based extensions. However, we cannot freely choose them. We have the following observation to prove Lemma \ref{lem:ns_cond}.

\begin{proposition} 
\label{prop:nec}
For any $I(f)$ interval-anchor assignment with an anchor matrix $A \in \mathbb{R}^{I(f)\times n}$, $a_{ij} > a_{(i-1)j}$ for all $i \in \{2,\dots,I(f)\}$ and $j \in \{1,\dots,n\}$.
\end{proposition}

It is possible to simplify the necessary and sufficient conditions given in Lemma \ref{lem:ns_cond}. We prove that looking at $(k,i)$ pairs in the form $(i-1,i)$ still provides the necessary and sufficient information compared to all $(k,i)$ pairs such that $k < i$.

\begin{lemma}
\label{lem:ns_cond_refined}
Let $f$ be an $n$-input symmetric Boolean function and $A \in \mathbb{R}^{I(f)\times n}$ be an anchor matrix for an $I(f)$ interval-anchor assignment. Then, the following condition is necessary and sufficient for any $i \in \{2,\dots,I(f)\}$:
\begin{align}
    \max_{X \in \{0,1\}^n : |X| = I_{i-1}} &\sum_{j=1}^n (a_{ij} - a_{(i-1)j})x_j < \frac{1}{2}\sum_{j=1}^n (a_{ij}^2 - a_{(i-1)j}^2) \nonumber \\
    \label{eq:ns_cond_refined}
    &< \min_{X \in \{0,1\}^n : |X| = I_{i-1}+1} \sum_{j=1}^n (a_{ij} -a_{(i-1)j})x_j
\end{align}
\end{lemma}

We can now prove Theorem \ref{th:upper} by finding explicit sets of anchors satisfying Lemma \ref{lem:ns_cond_refined}. In Theorem \ref{th:upper}, the $B$ matrix is actually constructed by $b_{ij} = a_{ij} - a_{(i-1)j}$ for $i \in \{2,\dots,I(f)\}$ and $j \in \{1,\dots,n\}$. Intuitively, as long as the entries of $B$ do not differ much, we can find anchors where the LHS and the RHS in Eq. \eqref{eq:ns_cond_refined} should hold. Therefore, $B = \mathds{1} + \epsilon M$ is a valid choice for sufficiently small $\epsilon > 0$.

In Theorem \ref{th:upper}, we remark that the full row rank $B$ matrix property is only sufficient to find a construction. For example, if $b_{ij} = 1/n$ for all $i,j$, this construction can be reduced to the PARITY-based construction given in \cite{hajnal2022nearest}.

In general, the PARITY-based approach results in $O(\log{n})$ resolution. We can pick $B$ such that any symmetric Boolean function can also have $O(\log{n})$ resolution with $I(f)$ anchors. We present a family of examples of $B$ with $O(\log{n})$ resolution in Section \ref{sec:resolution}. Such an example was already given in Appendix \ref{apx:counterexample}.

\section{Lower Bounds on the NN Complexity of Symmetric Boolean Functions}
\label{sec:lower_bound}

There are various circuit complexity lower bounds on symmetric Boolean functions \cite{paturi1990threshold,siu1991depth,spielman1993computing}. Similarly, lower bounds for the number of anchors can be proven for PARITY using $\{1,2\}$-sign representations of Boolean functions \cite{hajnal2022nearest,hansen2015polynomial}. We prove a more general result for periodic symmetric Boolean functions.

\begin{theorem}
    \label{th:per_lb}
    For a periodic symmetric Boolean function of $I(f)$ intervals, $NN(f) \geq I(f)$.
\end{theorem}

To prove this Theorem, we use a necessary condition for any NN representation of a periodic symmetric Boolean function.

\begin{proposition}
    \label{prop:lower_bound}
    Consider an NN representation of a periodic symmetric Boolean function $f(X)$ of period $T$ with an anchor matrix $A \in \mathbb{R}^{m\times n}$. Assume that $\sum_{j=1}^T a_{1j} \geq \sum_{j=1}^T a_{ij}$ for any $i \geq 2$. Then, no binary vectors that are closest to $a_1$ can have all 0s in the first $T$ coordinates.
\end{proposition}

We first claim that the assumption in Proposition \ref{prop:lower_bound} can be assumed without loss of generality by first observing that the NN representations of symmetric Boolean functions are equivalent up to permutations of anchors. We prove this explicitly in the Appendix \ref{apx:proofs_sec_lower}. Then, we find the anchor with the maximal $T$ sum in the entries and rearrange the rows and columns so that the claim $\sum_{j=1}^T a_{1j} \geq \sum_{j=1}^T a_{ij}$ for any $i \geq 2$ holds.

In addition, we prove that if $I(f) \leq 4$, then the NN complexity of any $n$-input symmetric Boolean function has the lower bound $NN(f) \geq I(f)$. We present the proof in Appendix \ref{apx:ex_no_interval_anchor} based on the tools developed in Section \ref{sec:resolution}.

\section{The Resolution of NN Representations}
\label{sec:resolution}

It is known that the PARITY-based approach results in $\lceil\log_2{(n+1)}\rceil$ resolution constructions for symmetric Boolean functions. We claim that the construction presented in Section \ref{sec:upper_bound} can admit $O(\log{n})$ resolution. Let Moore-Penrose pseudoinverse of a matrix $A$ be denoted by $A^+$.

\begin{theorem}
\label{th:sym_res}
    Suppose that we are given a rational matrix $B \in \mathbb{Q}^{m\times n}$ such that $ B = \mathds{1}_{m\times n} + \epsilon I_{m,n}$ where $\mathds{1}_{m\times n}$ is an all-one matrix, $\epsilon$ is a rational constant, and $I_{m,n}$ is a submatrix of the $n\times n$ identity matrix with the first $m$ rows. Then, $RES(B^+) = O(\log{n})$.
\end{theorem}

\begin{corollary}
    \label{cor:res}
     For every $n$-input symmetric Boolean function, there is an anchor matrix $A \in \mathbb{Q}^{I(f) \times n}$ such that $RES(A) = O(\log{n})$.
\end{corollary}

We refer to Theorem 3 in \cite{meyer1973generalized} to find the entries of $B^+$ explicitly and to prove Theorem \ref{th:sym_res}.

Now we consider the resolution of the NN representation of linear threshold functions. By a hyperplane argument, we prove the following result which can be verified algebraically.

\begin{theorem}
    \label{th:2-anchor}
    Let $f$ be a non-constant $n$-input linear threshold function with weight vector $w \in \mathbb{Z}^n$ and the threshold term $b \in \mathbb{Z}$. Then, there is a $2$-anchor NN representation with resolution $O(RES(w))$. In general, the resolution is $O(n\log{n})$.
\end{theorem}

\begin{corollary}
    \label{cor:2-anchor}
    Let $f$ be an $n$-input symmetric linear threshold function with threshold $b$. Then, there is a $2$-anchor NN representation with resolution $O(1)$.
\end{corollary}

Hence, symmetric linear threshold functions have $2$ intervals and optimal NN representations in size and resolution. It is easy to verify that
\begin{align}
        a_1 &= (0,\dots,0,-1/2,-1,\dots,-1) \\
        a_2 &= (2,\dots,2,\hphantom{-}3/2,\hphantom{-}1,\dots,\hphantom{-}1)
\end{align}
is always an NN representation where the number of $0$s in $a_1$ (and $2$s in $a_2$) is $b-1$.

However, if the number of intervals is at least 3, the resolution is lower bounded by $\Omega(\log{n})$ for some symmetric Boolean functions and the construction presented in Section \ref{sec:upper_bound} has optimal resolution by Corollary \ref{cor:res}.

\begin{theorem}
    \label{th:res_lb}
    Let $f$ be an $n$-input symmetric Boolean function where $f(X) = 1$ for $|X| =  \lfloor n/2 \rfloor + 1$ and $0$ otherwise. Then, any $3$-anchor NN representation has $\Omega(\log{n})$ resolution.
\end{theorem}

There is an anchor matrix $A$ for the function in Theorem \ref{th:res_lb} such that $RES(A) = O(\log{n})$ via PARITY-based extensions.

To prove Theorem \ref{th:res_lb}, we first prove that all but two symmetric Boolean functions with 3 intervals require an interval-anchor assignment. We then use the necessary and sufficient condition in Lemma \ref{lem:ns_cond_refined} to prove the resolution lower bound.

\begin{lemma}
    \label{lem:3_interval}
    Let $f$ be an $n$-input symmetric Boolean function with $3$ intervals. Then, any $I(f)$-anchor NN representation of $f$ needs to be an interval-anchor assignment except for the function $f(X) = 1$ (or $0$) for $|X| \in \{0,n\}$ and 0 (or $1$) otherwise.
\end{lemma}

An alternative NN representation for the function in Lemma \ref{lem:3_interval} is described in Appendix \ref{apx:ex_no_interval_anchor}. 

Finally, we use a circuit theoretic lower bound for a linear threshold function to give a lower bound on the NN representation size depending on resolution in Appendix \ref{apx:lt}.

\section{Concluding Remarks}

We study the information capacity of symmetric Boolean functions in the associative computation model. The information capacity consists of two quantities: NN complexity and resolution. Namely, the number of anchors in the NN representation and the number of bits required to represent the anchors. Specifically, we prove an upper bound on the NN complexity $NN(f) \leq I(f)$ and that this bound is tight for periodic symmetric Boolean functions. The upper bound of $I(f)$ on the NN complexity is based on a new NN construction with resolution $O(\log{n})$. We prove that our upper bound on resolution is tight for some symmetric Boolean functions. Surprisingly, to achieve optimality, the construction must be asymmetric in the entries of the anchors, contrary to the properties of symmetric Boolean functions. In addition, optimal constructions in size ($NN(f) = 2$) and resolution (a constant) are given for symmetric linear threshold functions.

For any symmetric Boolean function $f$, we conjecture that $NN(f) = I(f)$. While linear threshold functions have $NN(f) = 2$, the resolution is bounded by $O(n\log{n})$. An intriguing open problem is to find NN representations for all linear threshold functions of polynomial NN complexity and logarithmic (or constant) resolution.

\section*{Acknowledgement}
This research was partially supported by the Carver Mead New Adventure Fund.

\printbibliography

@article{hajnal2022nearest,
  title={Nearest Neighbor Representations of Boolean Functions},
  author={Hajnal, P{\'e}ter and Liu, Zhihao and Tur{\'a}n, Gy{\"o}rgy},
  journal={Information and Computation},
  volume={285},
  pages={104879},
  year={2022},
  publisher={Elsevier}
}

@article{haastad1995top,
  title={Top-down lower bounds for depth-three circuits},
  author={H{\aa}stad, Johan and Jukna, Stasys and Pudl{\'a}k, Pavel},
  journal={Computational Complexity},
  volume={5},
  number={2},
  pages={99--112},
  year={1995},
  publisher={Springer}
}

@article{stockmeyer1976combinational,
  title={On the combinational complexity of certain symmetric Boolean functions},
  author={Stockmeyer, Larry J},
  journal={Mathematical Systems Theory},
  volume={10},
  number={1},
  pages={323--336},
  year={1976},
  publisher={Springer}
}

@article{salzberg1995best,
  title={Best-case results for nearest-neighbor learning},
  author={Salzberg, Steven and Delcher, Arthur L. and Heath, David and Kasif, Simon},
  journal={IEEE Transactions on Pattern Analysis and Machine Intelligence},
  volume={17},
  number={6},
  pages={599--608},
  year={1995},
  publisher={IEEE}
}

@inproceedings{satoh1998analysis,
  title={Analysis of case-based representability of Boolean functions by monotone theory},
  author={Satoh, Ken},
  booktitle={International Conference on Algorithmic Learning Theory},
  pages={179--190},
  year={1998},
  organization={Springer}
}

@inproceedings{globig1996case,
  author    = {Christoph Globig and
               Steffen Lange},
  editor    = {Wolfgang Wahlster},
  title     = {Case-Based Representability of Classes of Boolean Functions},
  booktitle = {12th European Conference on Artificial Intelligence, Budapest, Hungary,
               August 11-16, 1996, Proceedings},
  pages     = {117--121},
  publisher = {John Wiley and Sons, Chichester},
  year      = {1996},
  bibsource = {dblp computer science bibliography, https://dblp.org}
}

@book{muroga1971threshold,
  title={Threshold logic and its applications},
  author={Muroga, Saburo},
  year={1971}
}

@inproceedings{muroga1959logical,
  title={Logical elements on majority decision principle and complexity of their circuit},
  author={Muroga, S},
  booktitle={Communications of the ACM},
  volume={2},
  number={7},
  pages={14--14},
  year={1959},
  organization={Assoc Computing Machinery 1515 Broadway, New York, NY 10036}
}

@article{minnick1961linear,
  title={Linear-input logic},
  author={Minnick, Robert C},
  journal={IRE Transactions on Electronic Computers},
  number={1},
  pages={6--16},
  year={1961},
  publisher={IEEE}
}

@article{cover1967nearest,
  title={Nearest neighbor pattern classification},
  author={Cover, Thomas and Hart, Peter},
  journal={IEEE transactions on information theory},
  volume={13},
  number={1},
  pages={21--27},
  year={1967},
  publisher={IEEE}
}

@inproceedings{paturi1990threshold,
  title={On threshold circuits for parity},
  author={Paturi, Ramamohan and Saks, Michael E},
  booktitle={Proceedings [1990] 31st Annual Symposium on Foundations of Computer Science},
  pages={397--404},
  year={1990},
  organization={IEEE}
}

@article{siu1991depth,
  title={Depth-size tradeoffs for neural computation},
  author={Siu, Kai-Yeung and Roychowdhury, Vwani P. and Kailath, Thomas},
  journal={IEEE Transactions on Computers},
  volume={40},
  number={12},
  pages={1402--1412},
  year={1991},
  publisher={IEEE Computer Society}
}

@article{spielman1993computing,
  title={Computing Arbitrary Symmetric Functions},
  author={Spielman, Daniel A},
  year={1992},
  publisher={Yale University. Department of Computer Science},
  url={https://cpsc.yale.edu/sites/default/files/files/tr906.pdf}
}

@inproceedings{wilfong1991nearest,
  title={Nearest neighbor problems},
  author={Wilfong, Gordon},
  booktitle={Proceedings of the seventh annual symposium on Computational Geometry},
  pages={224--233},
  year={1991}
}

@article{hansen2015polynomial,
  title={Polynomial threshold functions and Boolean threshold circuits},
  author={Hansen, Kristoffer Arnsfelt and Podolskii, Vladimir V},
  journal={Information and Computation},
  volume={240},
  pages={56--73},
  year={2015},
  publisher={Elsevier}
}

@article{hart1968condensed,
  title={The condensed nearest neighbor rule (corresp.)},
  author={Hart, Peter},
  journal={IEEE transactions on information theory},
  volume={14},
  number={3},
  pages={515--516},
  year={1968},
  publisher={Citeseer}
}

@article{meyer1973generalized,
  title={Generalized inversion of modified matrices},
  author={Meyer, Jr, Carl D},
  journal={SIAM journal on applied mathematics},
  volume={24},
  number={3},
  pages={315--323},
  year={1973},
  publisher={SIAM}
}

@article{roychowdhury1994lower,
  title={Lower bounds on threshold and related circuits via communication complexity},
  author={Roychowdhury, Vwani P and Orlitsky, Alon and Siu, Kai-Yeung},
  journal={IEEE Transactions on Information Theory},
  volume={40},
  number={2},
  pages={467--474},
  year={1994},
  publisher={IEEE}
}

@article{alon1997anti,
  title={Anti-Hadamard matrices, coin weighing, threshold gates, and indecomposable hypergraphs},
  author={Alon, Noga and V{\~u}, V{\u{a}}n H},
  journal={Journal of Combinatorial Theory, Series A},
  volume={79},
  number={1},
  pages={133--160},
  year={1997},
  publisher={Elsevier}
}

@article{haastad1994size,
  title={On the size of weights for threshold gates},
  author={H{\aa}stad, Johan},
  journal={SIAM Journal on Discrete Mathematics},
  volume={7},
  number={3},
  pages={484--492},
  year={1994},
  publisher={SIAM}
}

@article{siu1991power,
  title={On the power of threshold circuits with small weights},
  author={Siu, Kai-Yeung and Bruck, Jehoshua},
  journal={SIAM Journal on Discrete Mathematics},
  volume={4},
  number={3},
  pages={423--435},
  year={1991},
  publisher={SIAM}
}

@article{bruck1990harmonic,
  title={Harmonic analysis of polynomial threshold functions},
  author={Bruck, Jehoshua},
  journal={SIAM Journal on Discrete Mathematics},
  volume={3},
  number={2},
  pages={168--177},
  year={1990},
  publisher={SIAM}
}

@inproceedings{kilic2021neural,
  title={Neural Network Computations with {DOMINATION} Functions},
  author={Kilic, Kordag Mehmet and Bruck, Jehoshua},
  booktitle={2021 IEEE International Symposium on Information Theory (ISIT)},
  pages={1029--1034},
  year={2021},
  organization={IEEE}
}
\clearpage
\appendix  
\subsection{Proofs of the Theorems, Lemmas, and Propositions}
\label{apx:proofs}
\subsubsection{\textbf{\underline{Proofs for Section \ref{sec:upper_bound}}}}
\label{apx:proofs_sec_upper}
\begin{proof}[Proof of Proposition \ref{prop:nec}]
Consider two Boolean vectors $X = (x_1,\dots,0,\dots,x_n)$ and $X' = (x_1,\dots,1,\dots,x_n)$ where they only differ at $t^{th}$ location. Assume that this occurs at the boundary for an interval, i.e., $X$ and $X'$ are closer to $a_{i-1}$ and $a_i$ respectively. Then,
\begin{align}
    d(a_{i-1}, X)^2 - d(a_{i},X)^2 &< 0 \\
    d(a_{i-1}, X')^2 - d(a_{i},X')^2 &> 0
\end{align}
Subtracting both inequalities, we get $d(a_{i-1},X)^2 - d(a_{i-1},X')^2 - (d(a_i,X)^2 - d(a_i,X')^2) < 0$. Since $X$ and $X'$ differ only at the $t^{th}$ location, we get $(2a_{(i-1)t} - 1) - 2(a_{it} - 1) < 0$, hence, $a_{it} > a_{(i-1)t}$.
\end{proof}
\begin{proof}[Proof of Lemma \ref{lem:ns_cond}]
We begin by writing the simplest necessary and sufficient condition for an $I(f)$ interval-anchor assignment. Consider any two anchors $a_i$ and $a_k$ and assume that $|X| \in [I_{i-1}+1,I_{i}]$ so that $X$ is closer to $a_i$.
\begin{equation}
    d(a_i,X)^2 - d(a_k,X)^2 < 0
\end{equation}
which can be written as
\begin{equation}
    \sum_{j=1}^n (a_{ij}-x_j)^2 - \sum_{j=1}^n (a_{kj}-x_j)^2 < 0
\end{equation}
for $i \in \{2,\dots,I(f)\}$ and $k < i$. This implies
\begin{align}
    &\frac{1}{2}\sum_{j=1}^n (a_{ij}^2 - a_{kj}^2) \nonumber \\
    &\hphantom{aaaa}< \min_{X\in\{0,1\}^n : |X| \in [I_{i-1}+1 ,I_{i}]} \sum_{j=1}^n (a_{ij} - a_{kj})x_j
\end{align}
The RHS term is minimized when $|X| = \sum_{j=1}^n x_j = I_{i-1}+1$ by Proposition \ref{prop:nec} because each $a_{ij} - a_{kj} > 0$. Then,
\begin{equation}
    \label{eq:lem_1_rhs}
    \frac{1}{2}\sum_{j=1}^n (a_{ij}^2 - a_{kj}^2) < \min_{X\in\{0,1\}^n : |X| = I_{i-1}+1} \sum_{j=1}^n (a_{ij} - a_{kj})x_j
\end{equation}
We similarly do the analysis for the anchors $a_i$ and $a_k$ but this time assuming that $|X| \in [I_{k-1}+1,I_k]$ so that $X$ is closer to $a_k$. We still take $k < i$. Then,
\begin{align}
    \frac{1}{2}\sum_{j=1}^n (a_{kj}^2 - a_{ij}^2) &< \min_{X\in\{0,1\}^n : |X| \in [I_{k-1}+1 ,I_{k}]} \sum_{j=1}^n (a_{kj} - a_{ij})x_j \\
    \label{eq:lem_1_lhs}
    \frac{1}{2}\sum_{j=1}^n (a_{ij}^2 - a_{kj}^2) &> \max_{X\in\{0,1\}^n : |X| \in [I_{k-1}+1 ,I_{k}]} \sum_{j=1}^n (a_{ij} - a_{kj})x_j \nonumber \\
    &\hphantom{aa}= \max_{X\in\{0,1\}^n : |X| = I_k} \sum_{j=1}^n (a_{ij} - a_{kj})x_j 
\end{align}
Now, the RHS is maximized when $|X| = I_k$ by Proposition \ref{prop:nec}. By combining \eqref{eq:lem_1_rhs} and \eqref{eq:lem_1_lhs}, we complete the proof.
\end{proof}
\begin{proof}[Proof of Lemma \ref{lem:ns_cond_refined}]
    It is obvious that this is necessary if we apply $k = i-1$. We will show that we can use the Eq. \eqref{eq:ns_cond_refined} to obtain Lemma \ref{lem:ns_cond} by summing them telescopically until some fixed $k < i$.
\begin{align*}
    &\max_{X \in \{0,1\}^n : |X| = I_{i-1}} \sum_{j=1}^n (a_{ij} - a_{(i-1)j})x_j  \\
    &\hphantom{aaaaa}< \frac{1}{2}\sum_{j=1}^n (a_{ij}^2 - a_{(i-1)j}^2)  \\
    &\hphantom{aaaaa}< \min_{X \in \{0,1\}^n : |X| = I_{i-1}+1} \sum_{j=1}^n (a_{ij} -a_{(i-1)j})x_j \\
    &\hphantom{aaaaaaaaaa}\vdots \\
    &\max_{X \in \{0,1\}^n : |X| = I_{k}} \sum_{j=1}^n (a_{(k+1)j} - a_{kj})x_j \\
    &\hphantom{aaaa}< \frac{1}{2}\sum_{j=1}^n (a_{(k+1)j}^2 - a_{kj}^2) \\ 
    &\hphantom{aaaa}< \min_{X \in \{0,1\}^n : |X| = I_{k}+1} \sum_{j=1}^n (a_{(k+1)j} -a_{kj})x_j 
\end{align*}
Let us combine the inequalities and focus on the RHS.
\begin{equation}
     \sum_{l=k+1}^{i}
    \min_{X \in \{0,1\}^n : |X| = I_{l-1}+1} \sum_{j=1}^n (a_{lj} -a_{(l-1)j})x_j 
\end{equation}
We can replace the constraint sets of all of the minimization expressions with $I_{i-1} + 1$ because
\begin{align}
    &\min_{X \in \{0,1\}^n : |X| = I_{l-1}+1} \sum_{j=1}^n (a_{lj} -a_{(l-1)j})x_j\nonumber \\
    &\hphantom{aaaaa}< \min_{X \in \{0,1\}^n : |X| = I_{i-1}+1} \sum_{j=1}^n (a_{lj} -a_{(l-1)j})x_j
\end{align}
by Proposition \ref{prop:nec} and $l \leq i$. Then, we combine the objective functions telescopically again so that
\begin{equation}
    \frac{1}{2}\sum_{j=1}^n (a_{ij}^2 - a_{kj}^2) <
    \min_{X \in \{0,1\}^n : |X| = I_{i-1}+1} \sum_{j=1}^n (a_{ij} -a_{kj})x_j
\end{equation}
The LHS can be handled in a similar manner. Then, we can obtain Lemma \ref{lem:ns_cond} exactly.
\end{proof}
\begin{proof}[Proof of Theorem \ref{th:upper}]
Let us define a matrix $B \in \mathbb{R}^{I(f)-1 \times n}$ where $b_{ij} = a_{(i+1)j} - a_{ij}$. Rewriting the condition in Lemma \ref{lem:ns_cond_refined} using the $B$ matrix and transforming $i$ to $i+1$, we get
\begin{align}
    \label{eq:th_ns_cond}
    \max_{X \in \{0,1\}^n : |X| = I_{i}} &\sum_{j=1}^n b_{ij}x_j \nonumber \\
    &< \frac{1}{2}\sum_{j=1}^n \Bigg( -b_{ij}^2 + 2b_{ij}\sum_{k=1}^{i}b_{kj} + 2b_{ij}a_{1j}\Bigg) \nonumber \\
    &\hphantom{aaaa}< \min_{X \in \{0,1\}^n : |X| = I_{i}+1} \sum_{j=1}^n b_{ij}x_j
\end{align}
where the middle term is computed by the identity $a_{ij} = a_{1j} + \sum_{k=1}^{i-1} b_{kj} = a_{1j} -b_{ij} + \sum_{k=1}^i b_{kj}$.
\begin{align}
    \frac{1}{2}\sum_{j=1}^n (&a_{(i+1)j}^2 - a_{ij}^2) = \frac{1}{2}\sum_{j=1}^n b_{ij}(b_{ij}+2a_{ij}) \\
    &= \frac{1}{2}\sum_{j=1}^n\Bigg(b_{ij}^2 + 2b_{ij}\Big(a_{1j} - b_{ij} + \sum_{k=1}^{i} b_{kj}\Big)\Bigg) \\
    &= \frac{1}{2}\sum_{j=1}^n\Bigg(-b_{ij}^2 + 2b_{ij}\sum_{k=1}^{i} b_{kj} + 2b_{ij}a_{1j}\Bigg)
\end{align}
Let us take a convex combination of the LHS and RHS with some $\lambda_i \in (0,1)$ to make it equal to the middle term for each $i \in \{1,\dots,I(f)-1\}$. Therefore, we want to solve $a_1 = (a_{11},\dots,a_{1n})$ for the $Ba_1 = c$ where
\begin{align}
    \label{eq:bac}
    (Ba_1)_i = c_i = &\frac{1}{2}\sum_{j=1}^n b_{ij}^2 - \sum_{j=1}^n b_{ij}\sum_{k=1}^{i}b_{kj} \nonumber \\ 
    &+ \lambda_i \max_{X\in\{0,1\}^n : |X|=I_{i}} \sum_{j=1}^n b_{ij}x_j \nonumber \\
    &+ (1-\lambda_i)\min_{X\in\{0,1\}^n : |X|=I_{i}+1} \sum_{j=1}^n b_{ij}x_j
\end{align}
As long as this system of linear equations is consistent, we have a solution for $a_1$ and a construction for $I(f)$ interval-anchor assignment. If $B$ is a full row rank matrix, the system is guaranteed to be consistent. As long as the LHS is smaller than the RHS in Eq. \eqref{eq:th_ns_cond}, any full row rank $B$ matrices will work. Therefore, we can claim that $B = \mathds{1} + \epsilon M$ for sufficiently small $\epsilon > 0$ and full row rank $M$ can be used to construct an anchor matrix for an arbitrary symmetric Boolean function.
\end{proof}
\subsubsection{\textbf{\underline{Proofs for Section \ref{sec:lower_bound}}}}
\label{apx:proofs_sec_lower}
\begin{proposition}
    \label{prop:sym_nn_equivalence}
    Suppose that an anchor matrix $A \in \mathbb{R}^{m\times n}$ is an NN representation with $m$ anchors for an $n$-input symmetric Boolean function. Then, any permutation of rows and columns of $A$ is an NN representation for the same function.
\end{proposition}
\begin{proof}[Proof of Proposition \ref{prop:sym_nn_equivalence}]
    Permutation of rows is trivial and can be done for any NN representation. For the columns, we can use the definition of symmetric Boolean functions where $f(X) = f(\sigma(X))$ for any permutation $\sigma(.)$. The closest anchor to a Boolean vector $X$ can be found by
    \begin{align}
        \arg\min_{i} d(a_i,X)^2 &= \arg\min_{i} |X| - 2(AX)_i + ||a_i||_2^2 \\
        &= \arg\max_{i} \Big(2AX - \diag(AA^T)\Big)_i
    \end{align}
where $||.||_2$ denotes the Euclidean norm and $\diag(M)$ is the all-zero matrix except the diagonal entries of $M$. Let $P \in \{0,1\}^{n\times n}$ be a permutation matrix. Then,
    \begin{equation}
        \label{eq:permutation_eq}
        \arg\max_{i} \Big(2A(PX) - \diag(AA^T)\Big)_i
    \end{equation}
is an anchor index assigned to the same anchor type (either positive or negative) because $f$ is symmetric. Note that Eq. \eqref{eq:permutation_eq} is equivalent to
    \begin{equation}
    \begin{split}
        \arg\max_{i} \Big(2(AP)X &- \diag((AP)(AP)^T)\Big)_i \\
        &= \arg\max_{i} \Big(2(AP)X - \diag(AA^T)\Big)_i
    \end{split}
    \end{equation}
\end{proof}
\begin{proof}[Proof of Proposition \ref{prop:lower_bound}]
    We use a similar idea used for Proposition \ref{prop:nec}. Assume that there is a vector $X = (0,\dots,0,x_{T+1},\dots,x_n)$ assigned to $a_1$ for contradiction. Let $X' = (1,\dots,1,x_{T+1},\dots,x_n)$. Then,
    \begin{align}
        \label{eq:first_cond_lb}
        d(a_1,X)^2 &< d(a_i,X)^2 \hphantom{'aa} \forall i \geq 2 \\
        \label{eq:sec_cond_lb}
        d(a_1,X')^2 &> d(a_i,X')^2 \hphantom{aa} \text{for some } i \geq 2
    \end{align}
    We get Eq. \eqref{eq:first_cond_lb} by the assumption for contradiction. Eq. \eqref{eq:sec_cond_lb} is obtained by the fact that $f(X) \neq f(X')$ because we have a jump of length $T$ for a given value of $|X|$. Subtracting both, we get
    \begin{equation}
        \sum_{j=1}^T a_{1j} < \sum_{j=1}^T a_{ij}
    \end{equation}
    for some $i \geq 2$, which is the desired contradiction.
\end{proof}
\begin{proof}[Proof of Theorem \ref{th:per_lb}]
    Let $f$ be a periodic symmetric Boolean function with period $T$ so that $n = I(f)T$. Suppose that the NN representation for this function has an anchor matrix $A \in \mathbb{R}^{m \times n}$ where
    \begin{equation}
        \sum_{j=(k-1)T+1}^{kT} a_{kj} \geq \sum_{j=(k-1)T+1}^{kT} a_{ij}
    \end{equation}
    for any $k \in \{1,\dots,m-1\}$ and $k < i$. This holds without loss of generality because we can rank the maximal sums of the $k$ entries of the anchors and rearrange the rows and columns of $A$ by Proposition \ref{prop:sym_nn_equivalence}. 
    
    Iteratively, for each $k \in \{1,\dots,m-1\}$, we see that if $X_i = 0$ for $i \in \{1,\dots,kT\}$, then $(a_1,\dots,a_k)$ cannot be assigned to $X$ by Proposition \ref{prop:lower_bound}. Since $n = I(f)T$ for a periodic symmetric Boolean function, if $m < I(f)$, there will remain $X$ vectors with different $f(X)$ values assigned to a single anchor, leading to a contradiction.
\end{proof}
\subsubsection{\textbf{\underline{Proofs for Section \ref{sec:resolution}}}}
\label{apx:proofs_sec_res}
\begin{theorem}[\cite{meyer1973generalized}]
    \label{th:pseudo}
    Suppose that we are given a matrix $A \in \mathbb{R}^{m\times n}$ and column vectors $c \in \mathbb{R}^m$,$d \in \mathbb{R}^n$. Then, if $c$ is in the column space of $A$ and the quantity $\beta = 1 + d^T A^+ c \neq 0$, we have
    \begin{align}
        (A + cd^T)^+ = A^+ &+ \frac{1}{\beta} v^T k^T A^+ \nonumber \\
        &- \frac{(||k||^2 v^T + k\beta) (||v||^2 k^T A^+ + h\beta)}{\beta (||k||^2||v||^2 + \beta^2)}
    \end{align}
    where $v = d^T(I - A^+A)$, $k = A^+ c$, and $h = d^T A^+$.
\end{theorem}
\begin{proof}[Proof of Theorem \ref{th:sym_res}]
    To apply Theorem \ref{th:pseudo}, we pick $c = \mathds{1}_{m}$ and $d = \mathds{1}_n$ where $\mathds{1}_k$ denotes the all-one vector in $\mathbb{R}^k$ for an integer $k$. Also, pick $A = \epsilon I_{m,n}$. Therefore, $A^+ = \frac{1}{\epsilon}I_{m,n}^T$.
    
    We compute the pseudoinverse of $\mathds{1}_{m\times n} + \epsilon I_{m,n}$ as the following closed form expression.
    \begin{equation}
        \label{eq:closed_from_pseudo}
        \Big(\mathds{1}_{m\times n} + \epsilon I_{m,n}\Big)^+ =
        \begin{bmatrix}
            \frac{1}{\epsilon} I_{m\times m} - \frac{1 + n/\epsilon}{m(n-m)+(m+\epsilon)^2}\mathds{1}_{m\times m} \\
            \frac{1}{m(n-m)+(m+\epsilon)^2}\mathds{1}_{(n-m)\times m}
        \end{bmatrix}
    \end{equation}
    Since $m \leq n$ and $\epsilon$ is a rational constant, the denominator of all entries of the pseudoinverse is a function of $O(n^2)$ (and $O(n)$ for the numerator). Therefore, the resolution of the entries of $B^+$ is $O(\log{n})$.
\end{proof}
\begin{proof}[Proof of Corollary \ref{cor:res}]
    For the matrix $B = \mathds{1}_{m\times n} + \epsilon M$ in Theorem \ref{th:upper} with $m = I(f)-1$, we pick $\epsilon = 1/2$ and $M = I_{m,n}$ where $I_{m,n}$ is a submatrix of $n\times n$ identity matrix with the first $m$ rows. It can be easily verified that this $B$ matrix where $b_{ij} = a_{ij} - a_{(i-1)j}$ satisfies the necessary and sufficient condition \eqref{eq:ns_cond_refined} for the interval-anchor assignment. The first anchor $a_1$ can be obtained by $B^+c$ where $c$ is given in Eq. \eqref{eq:bac}.
    
    With $\lambda_i = 1/2$ for all $i \in \{1,\dots,m\}$, each entry of $c$ depends on the entries on $B$ and $n$ polynomially. By Theorem \ref{th:sym_res}, $B^+$ has entries with at most $O(\log{n})$ resolution and therefore, we see that the entries of $a_1 = B^+c$ has at most $O(\log{n})$ resolution (consequently, $a_2,\dots,a_m$). In conclusion, we can always obtain an interval-anchor assignment with $O(\log{n})$ resolution.
\end{proof}
\begin{proof}[Proof of Theorem \ref{th:2-anchor}]
    Geometrically, any $2$-anchor NN representation construction of linear threshold functions can be written in the following form for a real number $c > 0$ and an arbitrary $x^* \in \mathbb{R}^n$ such that $w^T x^* = b'$ where $b' = b - 0.5$.
    \begin{align}
        a_1 &= x^* - cw \\
        a_2 &= x^* + cw
    \end{align}
    This can be seen algebraically as well. We perturb $b$ to $b'$ to consider the points on the hyperplane itself. We first claim that there exists $x^*$ such that $RES(x^*) \leq RES(w)$.
    
    Since $f(X)$ is not constant, there is always a pair of binary vectors $X'$ and $X''$ such that $w^T X' < b' < w^T X''$ where $X_i' = X_i''$ for all $i \in \{1,\dots,n\}$ except for a unique $i = k \in \mathbb{Z}$. Then, we construct
    \begin{align}
        x_i^* &= X_i' \text{ for } i \neq k \\
        x_k^* &= \frac{b' - w^T X'}{w_k}
    \end{align}
    Clearly, $x_i^*$s are binary except $i = k$ where $|b' - w^TX'| < w_k$ and therefore, $RES(x_k^*) = \lceil\log_2{w_k+1}\rceil$. In conclusion, $RES(x^*) = \lceil\log_2{w_k+1}\rceil \leq RES(w)$.

    Picking $c = 1$, we see that $RES(A) = O(RES(w))$. Moreover, since $w_i = 2^{O(n\log{n})}$ for $i \in \{1,\dots,n\}$ in general \cite{alon1997anti,haastad1994size,muroga1971threshold}, we conclude that $RES(A) = O(n\log{n})$.
\end{proof}
\begin{proof}[Proof of Theorem \ref{th:res_lb}]
For the $n$-input symmetric Boolean function given in the Theorem, we always have an interval-anchor assignment by Lemma \ref{lem:3_interval}. For an interval-anchor assignment, let $B \in \mathbb{R}^{2\times n}$ $b_{ij} = a_{(i+1)j}-a_{ij} > 0$ for $i \in \{1,2\}$ and $j \in \{1,\dots,n\}$. Then, let us write the necessary and sufficient conditions by Lemma \ref{lem:ns_cond_refined}. Without loss of generality, assume that $0 < b_{11} \leq b_{12} \leq \dots \leq b_{1n}$ as we can always reorder the columns of $B$. Then, the first condition will be
\begin{align}
    &\max_{X\in\{0,1\}^n:|X| = \lfloor n/2 \rfloor} \sum_{j=1}^n b_{1j}x_j < \frac{1}{2} \sum_{j=1}^n b_{1j}(b_{1j} + 2a_{1j})\nonumber \\
    \label{eq:res_nec}
    &\hphantom{aaaaaaaaaaa}< \min_{X\in\{0,1\}^n:|X| = \lfloor n/2 \rfloor + 1} \sum_{j=1}^n b_{1j}x_j 
\end{align}
We have two important inequalities based on this condition.
\begin{align}
    \label{eq:first_b_bound}
b_{1(n-\sqrt{n})} &< b_{1(\sqrt{n}+1)} + \frac{1}{\sqrt{n}}b_{11} \\
    \label{eq:second_b_bound}
b_{1n} &< 2b_{1\sqrt{n}} \leq 2b_{1(\sqrt{n}+1)}
\end{align}
To prove Eq. \eqref{eq:first_b_bound}, we assume the contrary such that $b_{1(n-\sqrt{n})} \geq b_{1(\sqrt{n}+1)} + \frac{1}{\sqrt{n}}b_{11}$. Also, $b_{1(n-\sqrt{n}+j)} \geq b_{1(n-\sqrt{n})}$ for $1 \leq j \leq \sqrt{n}$ and $b_{1(\sqrt{n}+1)} \geq b_{1j}$ for $1 \leq j \leq \sqrt{n}+1$. Hence, if we sum over $j \in \{1,\dots,\sqrt{n}\}$, we get
\begin{equation}
    \sum_{j=1}^{\sqrt{n}} b_{1(n-\sqrt{n}+j)} \geq \sum_{j=1}^{\sqrt{n}} \Bigg(b_{1(j+1)} + \frac{1}{\sqrt{n}} b_{11}\Bigg)
\end{equation} 
Let us add $b_{1j}$ for $j \in \{n - \lfloor n/2 \rfloor + 1,\dots,n-\sqrt{n}\}$ to both sides. Then, we get
\begin{equation}
    \sum_{j=n - \lfloor n/2 \rfloor + 1}^{n} b_{1j} \geq \sum_{j=1}^{\sqrt{n}} \Bigg(b_{1(j+1)} + \frac{1}{\sqrt{n}} b_{11}\Bigg) + \sum_{j=n - \lfloor n/2 \rfloor + 1}^{n-\sqrt{n}} b_{1j}
\end{equation}
where the LHS has $\lfloor n/2 \rfloor$ many terms and the RHS has $\lfloor n/2 \rfloor + 1$ many terms. This contradicts Eq. \eqref{eq:res_nec}.

Similarly, to prove Eq. \eqref{eq:second_b_bound}, we assume the contrary such that $b_{1n} \geq 2b_{1\sqrt{n}} \geq b_{11} + b_{12}$. If we add both sides $\sum_{i=n - \lfloor n/2 \rfloor + 1}^{n-1} b_{1i}$, we get 
\begin{equation}
    \sum_{i=n - \lfloor n/2 \rfloor + 1}^{n} b_{1i} \geq b_{11} + b_{12} + \sum_{i=n - \lfloor n/2 \rfloor + 1}^{n-1} b_{1i} 
\end{equation}
where the LHS has $\lfloor n/2 \rfloor$ many terms and the RHS $\lfloor n/2 \rfloor +1$ many terms, contradicting Eq. \eqref{eq:res_nec}.

We now want to bound the middle term in Eq. \eqref{eq:res_nec}. We divide the sum in three parts. Essentially, we want to show that the main contribution in the value of the whole summation is due to the middle term.
\begin{align}
    \sum_{j=1}^{\sqrt{n}} b_{1j}(b_{1j}+2a_{1j}) &+ \sum_{j=\sqrt{n}+1}^{n-\sqrt{n}} b_{1j}(b_{1j}+2a_{1j}) + \nonumber\\
    \label{eq:three_sum_part}
    &\hphantom{aaaa} \sum_{j=n-\sqrt{n}+1}^{n} b_{1j}(b_{1j}+2a_{1j})
\end{align}
We also divide the middle term in Eq. \eqref{eq:three_sum_part} into two parts depending on whether the $b_{1j}+2a_{1j}$ terms are positive or not. Let $\mathcal{J}^+$ denote the indices $j$ where $b_{1j}+2a_{1j} > 0$ and $\mathcal{J}^-$ otherwise. Then, we define $\mathcal{S}^+ = \sum_{j \in \mathcal{J}^+} (b_{1j}+2a_{1j})$
and $\mathcal{S}^- = \sum_{j \in \mathcal{J}^-} (b_{1j}+2a_{1j})$ so that 
\begin{align}
    \label{eq:mid_term_bound}
    \sum_{j=\sqrt{n}+1}^{n-\sqrt{n}} &b_{1j}(b_{1j}+2a_{ij}) > b_{1(\sqrt{n}+1)}\mathcal{S}^+ + b_{1(n-\sqrt{n})}\mathcal{S}^-  \\
    \label{eq:mid_term_bound_2}
    &>b_{1(\sqrt{n}+1)}\Big(\mathcal{S}^+ + \mathcal{S}^-\Big) - \sqrt{n}2^{r+1}b_{1(\sqrt{n}+1)}
\end{align}
where $r$ denotes the resolution of the representation. If $r$ is the resolution of the representation (and hence, the anchor matrix $A$), it is clear that $|b_{1j}+2a_{1j}| = |a_{1j} + a_{2j}| \leq |a_{1j}| + |a_{2j}| < 2^{r+1}$ and we can rewrite Eq. \ref{eq:first_b_bound} so that $b_{1(n-\sqrt{n})} < b_{1(\sqrt{n}+1)}\Big(1+\frac{1}{\sqrt{n}}\Big)$. We also have the loose bound $\mathcal{S}^- > -n2^{r+1}$ and by using these, we can obtain Eq. \eqref{eq:mid_term_bound_2}.

We find lower bounds for the first term in Eq. \eqref{eq:three_sum_part} by $-b_{1(\sqrt{n}+1)}2^{r+1}\sqrt{n}$ and the third term by $-b_{1n}2^{r+1}\sqrt{n} > -2\sqrt{n}b_{1(\sqrt{n}+1)}2^{r+1}$.

We combine everything and the upper bound in Eq.\eqref{eq:res_nec} should hold for the expression that we obtain. We again use Eq. \eqref{eq:first_b_bound}.
\begin{align}
    b_{1(\sqrt{n}+1)}&\Big(\mathcal{S}^+ + \mathcal{S}^- - \sqrt{n}2^{r+3}\Big) \nonumber \\ 
    &\hphantom{aaa}< 2\sum_{j=1}^{\lfloor n/2\rfloor +1} b_{1j} < 2\sum_{j=1}^{\lfloor n/2\rfloor + 1} b_{1(n-\sqrt{n})} \nonumber \\ 
    &\hphantom{aaa} < 2b_{1(\sqrt{n}+1)} \Big(\lfloor n/2 \rfloor +1\Big)\Big(1+\frac{1}{\sqrt{n}}\Big)\\
    \mathcal{S}^+ + \mathcal{S}^- &< n + O(\sqrt{n}) + \sqrt{n}2^{r+3}
\end{align}

Let us rewrite $\sum_{j=1}^n (a_{1j} + a_{2j}) = \sum_{j=1}^n (b_{1j} + 2a_{1j}) =  \mathcal{S}^+ + \mathcal{S}^- + \sum_{j=1}^{\sqrt{n}} (b_{1j}+2a_{1j}) + \sum_{j=n-\sqrt{n}+1}^{n} (b_{1j}+2a_{1j})$. Then,
\begin{align}
    \label{eq:s_bound}
    \sum_{j=1}^n (a_{1j} + a_{2j}) < \Big(\mathcal{S}^+ + \mathcal{S}^- + 2\sqrt{n}2^{r+1}\Big)
\end{align}
by $|b_{1j}+2a_{1j}| < 2^{r+1}$. Let $r < c\log_2{n}$ for some constant $0 < c < 1/2$. Clearly, the RHS of Eq.\eqref{eq:s_bound} is less than $n + O(n^{\epsilon})$ where $1/2 < \epsilon < 1$ is a constant.

Similarly, we obtain the corresponding lower bound and use the other necessary and sufficient condition to prove another inequality corresponding to Eq. \eqref{eq:nec_suf_2}. For a constant $1/2 < \epsilon < 1$, we have
\begin{align}
    \label{eq:nec_suf_1}
    n - O(n^\epsilon) &< \sum_{j=1}^n (a_{1j}+a_{2j}) < n + O(n^\epsilon) \\
    \label{eq:nec_suf_2}
    n - O(n^\epsilon) &< \sum_{j=1}^n (a_{3j} + a_{2j}) < n + O(n^\epsilon)
\end{align}
Subtracting both, we get
\begin{align}
    -O(n^\epsilon) < \sum_{j=1}^n a_{3j} - a_{1j} < O(n^\epsilon)
\end{align}
and therefore, $0 < \frac{1}{2^{2r}} \leq |a_{3j}-a_{1j}| < \frac{O(n^\epsilon)}{n}$ for $i \in \{1,2\}$ and some $j \in \{1,\dots,n\}$. Hence, $RES(A) = r = \Omega(\log{n})$.
\end{proof}
\begin{proposition}
    \label{prop:convex}
    For an arbitrary NN representation of size at least $2$, pick any positive and negative anchor, $a$ and $b$. Let $\mathcal{X}$ and $\mathcal{Y}$ be the set of binary vectors closest to $a$ and $b$, respectively. Then, $conv(\mathcal{X}) \cap conv(\mathcal{Y}) = \emptyset$ where $conv(\mathcal{A})$ denotes the convex hull of a set $\mathcal{A}$.
\end{proposition}
\begin{proof}[Proof of Proposition \ref{prop:convex}]
    Let $|\mathcal{X}| = k$ and $|\mathcal{Y}| = l$. We assume that $\mathcal{X} = \{X_1,\dots,X_k\}$ and $\mathcal{Y} = \{Y_1,\dots,Y_l\}$. Also, let $u$ and $v$ be arbitrary convex combinations for the sets $\mathcal{X}$ and $\mathcal{Y}$. That is, $u = \sum_{i=1}^{k} \lambda_i X_i$ and $v = \sum_{i=1}^{l} \mu_i Y_i$ where $X_i \in \mathcal{X}$, $Y_i \in \mathcal{Y}$, $\lambda_i \in [0,1]$, $\sum_{i=1}^k\lambda_i = 1$, $\mu_i \in [0,1]$, and $\sum_{i=1}^l \mu_i = 1$.

    We know that $d(a,X_i)^2 < d(b,X_i)^2$ for $i \in \{1,\dots,k\}$ and $d(a,Y_i)^2 > d(b,Y_i)^2$ for $i \in \{1,\dots,l\}$. Suppose for contradiction that there is a set of $\lambda$s and $\mu$s so that $u = v$.
    \begin{align}
        d(a,u)^2 &- d(b,u)^2 = ||a||^2 - ||b||^2 - 2(a-b)^T \sum_{i=1}^k\lambda_i X_i \\
        &= \sum_{i=1}^k \lambda_i \Big(||a||^2 - 2(a-b)^T X_i -||b||^2\Big) \\
        &= \sum_{i=1}^k \lambda_i \Big(d(a,X_i)^2 - d(b,X_i)^2\Big) < 0
    \end{align}
    where we use $\lambda_i \geq 0$ and $\sum_{i=1}^k \lambda_i = 1$. We do the same for $v$ since $u = v$, to obtain
    \begin{equation}
        d(a,v)^2 - d(b,v)^2 = \sum_{i=1}^l \mu_i \Big(d(a,Y_i)^2 - d(b,Y_i)^2\Big) > 0
    \end{equation}
    resulting in a contradiction.
\end{proof}
\begin{proof}[Proof of Lemma \ref{lem:3_interval}]
    There are only two possible anchor assignments for symmetric Boolean functions with $3$ intervals besides an interval-anchor assignment: Either $a_2$ is assigned to the region $|X| \in [I_1 +1, I_2]$ and $a_1$ \& $a_3$ shares the rest (\textbf{Case 1}) or $a_2$ \& $a_3$ shares the same region $|X| \in [I_1+1,I_2]$ and $a_1$ is assigned to the rest (\textbf{Case 2}).
    
    \textbf{For Case 1:} We pick $X_1$ and $X_2$ assigned either both to $a_1$ or $a_3$ where $|X_1| = t_1 < I_1 + 1$ and $|X_2| = t_2 > I_2$. Consider also an integer $t \in [I_1+1,I_2]$ independently. This is always possible except the function given in the description. However, for that function, this case implies an interval-anchor assignment as there are unique vectors both for $|X| = 0$ and $|X| = n$. We denote the set of coordinates where $X_1 = X_2 = 1$ by $\mathcal{S}_1$, $X_1 = 1, X_2 =0$ by $\mathcal{S}_2$, $X_1 = 0, X_2 = 1$ by $\mathcal{S}_3$, and finally, $X_1 = X_2 = 0$ by $\mathcal{S}_4$. Clearly, $t_1 = |\mathcal{S}_1| + |\mathcal{S}_2|$, $t_2 = |\mathcal{S}_1| + |\mathcal{S}_3|$, $|\mathcal{S}_1| + |\mathcal{S}_2| + |\mathcal{S}_3| + |\mathcal{S}_4| = n$. Consider the following as an example.
    \begin{align}
        \label{eq:s_representation_1}
        X_1 = (1,\dots,1,1,\dots,1,0,\dots,0,0,\dots,0) \\
        \label{eq:s_representation_2}
        X_2 = (\underbrace{1,\dots,1}_{\mathcal{S}_1},\underbrace{0,\dots,0}_{\mathcal{S}_2},\underbrace{1,\dots,1}_{\mathcal{S}_3},\underbrace{0,\dots,0}_{\mathcal{S}_4})
    \end{align}
    Note that the example representation given in Eq. \eqref{eq:s_representation_1} and \eqref{eq:s_representation_2} can be assumed without loss of generality by the reordering of indices. Also, $\mathcal{S}_3 \neq \emptyset$ is necessary by the choice of $t_1,t_2$, and $t$. It is also clear that $|\mathcal{S}_3| = t_2 - t_1 + |\mathcal{S}_2| \geq t - t_1$ and $|\mathcal{S}_3| = t_2 - |\mathcal{S}_1| \geq t - |\mathcal{S}_1|$.
    
    It can be easily verified that the convex combination $X' = \frac{t_2 - t}{t_2 - t_1} X_1 + \frac{t-t_1}{t_2-t_1}X_2$ lies on the hyperplane $|X| = t$. We further claim that it is in the convex hull of the set of binary vectors $X$s where $|X| = t$.
    
    Let $Y$ be the average of all binary vectors on the hyperplane $|X| = t$ with $t_1$ many $1$s at the locations of $1$s of $X_1$, $S_4$ many $0$s at the location of $0$s of $X_1$. That is, we put the remaining $t-t_1$ many ones to the indices in $\mathcal{S}_3$. It is clear that there are $\binom{S_3}{t-t_1}$ many such vectors. The average value in the $\mathcal{S}_3$ region is  $K = \binom{|\mathcal{S}_3|-1}{t-t_1-1} / \binom{|\mathcal{S}_3|}{t-t_1} = (t-t_1)/|\mathcal{S}_3|$.

    Similarly, we define $Z$ as the average of all binary vectors on the hyperplane $|X| = t$ with all $1$s for the indices in $\mathcal{S}_1$ and the $t-|\mathcal{S}_1| = $ remaining $1$s will be distributed in the region $\mathcal{S}_3$. The average value in the $\mathcal{S}_3$ region is  $L = \binom{|\mathcal{S}_3|-1}{t-|\mathcal{S}_1|-1} / \binom{|\mathcal{S}_3|}{t-|\mathcal{S}_1|} = (t-|\mathcal{S}_1|)/|\mathcal{S}_3|$.
    \begin{align}
        Y &= (1,\dots,1,1,\dots,1,K,\dots,K,0\dots,0) \\
        Z &= (\underbrace{1,\dots,1}_{\mathcal{S}_1},\underbrace{0,\dots,0}_{\mathcal{S}_2},\underbrace{L,\dots,L}_{\mathcal{S}_3},\underbrace{0,\dots,0}_{\mathcal{S}_4})
    \end{align}
    Then, it is easy to verify that $X' = \frac{t_2 - t}{t_2 - t_1} Y + \frac{t-t_1}{t_2-t_1}Z$ as $K(t_2-t)/(t_2-t_1) + L(t-t_1)/(t_2-t_1) = (t-t_1)/(t_2-t_1)$.

    \textbf{For Case 2:} When $a_2$ \& $a_3$ share the interval $[I_1+1,I_2]$, all the vectors such that $|X| = t$ for $t \in [I_1+1,I_2]$ cannot be assigned to one of the anchors as this will reduce to \textbf{Case 1}. Therefore, $a_2$ \& $a_3$ must share $|X| = t$ for any $t \in [I_1+1,I_2]$.
    
    Let $X_1 \neq X_2$, and the $\mathcal{S}_i$ be the same as in \textbf{Case 1} except $|X_1| = |X_2| = I_1+1$. We can choose $X_1$ and $X_2$ closest to $a_2$ or $a_3$ respectively. In this case, $\mathcal{S}_2$ and $\mathcal{S}_3$ are not empty. We define $X_3$ for $X_1 \& X_2$ such that $|X_3| = I_2$ and $X_3$ has all $1$s for the indices in $\mathcal{S}_1$ and $\mathcal{S}_2$ except the entry in the first index for $\mathcal{S}_2$ is zero. Also, we keep the first entry in $\mathcal{S}_3$ to be zero as well and fill the rest of the indices arbitrarily. Hence, by construction, we have $I_2 < n-1$.
    \begin{align}
        X_3 = (\underbrace{1,\dots,1}_{\mathcal{S}_1},\underbrace{0,1,\dots,1}_{\mathcal{S}_2},\underbrace{0,1,0,\dots,1}_{\mathcal{S}_3 \cup \mathcal{S}_4})
    \end{align}

    Similarly, by reversing the roles of $0$s and $1$s, we can obtain $I_1 > 0$ condition. Therefore, the construction of $X_3$ is possible whenever either $I_1 > 0$ or $I_3 < n-1$. The only exception is when $f(X) = 1$ (or $0$) for $|X| \in \{0,n\}$ and $0$ (or $1$) otherwise.

    Notice that $X_3$ is constructed by using $\mathcal{S}$ depending both on $a_2$ and $a_3$. In general, $X_3$ can be assigned either to $a_2$ or $a_3$. Assume that it is assigned to $a_2$. Otherwise, change $X_1$ to $X_2$ in the following calaims: Let $Y = X_1 \land X_3$ and $Z = X_1 \lor X_3$. Clearly, $|Y| \leq I_1 - 1$ and $|Z| \geq I_2+1$ and they are both assigned to $a_1$. We finally claim that $Y/2 + Z/2 = X_1/2 + X_3/2$.
    
    In conclusion, both cases contradict Proposition \ref{prop:convex} and an interval-anchor assignment is necessary.
\end{proof}

\subsection{A Counterexample for the PARITY-based Extensions for NN Representations of Symmetric Boolean Functions}
\label{apx:counterexample}

Let $f(X)$ be the function given in Eq. \eqref{eq:counterexample}. For the sake of contradiction, we assume that there exist a PARITY-based extension to this symmetric function, namely, $a_{ij} = a_{ik}$ for any $i \in \{1,\dots,5\}$ and $j,k \in \{1,\dots,8\}$.

\begin{equation}
    \label{eq:counterexample}
    \begin{tabular}{c|c c c}
        $|X|$ & $f(X)$ \\
        \cline{1-2}
         0 & \textcolor{blue}{1} & \hphantom{a} & $I_{1} = 0$\\
         1 & \textcolor{red}{\textbf{0}} & \hphantom{a} & $I_{2} = 1$\\ 
         2 & \textcolor{blue}{1} & \hphantom{a} & \\
         3 & \textcolor{blue}{1} & \hphantom{a} & \\
         4 & \textcolor{blue}{1} & \hphantom{a} & \\
         5 & \textcolor{blue}{1} & \hphantom{a} & \\
         6 & \textcolor{blue}{1} & \hphantom{a} & $I_{3} = 6$\\
         7 & \textcolor{red}{\textbf{0}} & \hphantom{a} & $I_4 = 7$\\
         8 & \textcolor{blue}{1} & \hphantom{a} & $I_5 = 8$
    \end{tabular}
\end{equation}
We can rewrite the necessary and sufficient conditions given in Lemma \ref{lem:ns_cond_refined}.
We also know that $a_{ij} + a_{(i-1)j}$ (or $a_{ij} - a_{(i-1)j}$) has the same value for all $j \in \{1,\dots,8\}$ and $i \in \{2,\dots,5\}$. Then, we get
\begin{align}
    I_{i-1}\big(a_{i1} - a_{(i-1)1}\big) &< 4\big(a_{i1} - a_{(i-1)1}\big)\big(a_{i1} + a_{(i-1)1}\big)\nonumber \\
    &\hphantom{aaaa}< (I_{i-1}+1)\big(a_{i1} - a_{(i-1)1}\big) \\
    \frac{I_{i-1}}{4} &< \big( a_{i1} + a_{(i-1)1}\big) < \frac{I_{i-1}+1}{4}
\end{align}
Recall that Proposition \ref{prop:nec} still applies here. Therefore, $a_{1j} < a_{2j} < a_{3j} < a_{4j} < a_{5j}$ for all $j \in \{1,\dots,8\}$. More explicitly, we have the following system of inequalities
\begin{align}
    \label{eq:c_ex_nec}
    a_{11} &< a_{21} < a_{31} < a_{41} < a_{51} \\
    \label{eq:c_ex_1}
    0 &< a_{21} + a_{11} < \frac{1}{4} \\
    \label{eq:c_ex_2}
    \frac{1}{4} &< a_{31} + a_{21} < \frac{2}{4} \\
    \label{eq:c_ex_3}
    \frac{6}{4} &< a_{41} + a_{31} < \frac{7}{4} \\
    \label{eq:c_ex_4}
    \frac{7}{4} &< a_{51} + a_{41} < \frac{8}{4}
\end{align}
Firstly, we multiply Eq. \eqref{eq:c_ex_1} and \eqref{eq:c_ex_3} with $-1$ and sum all Eq. \eqref{eq:c_ex_1},\eqref{eq:c_ex_2},\eqref{eq:c_ex_3}, and \eqref{eq:c_ex_4}. Secondly, we multiply Eq. \eqref{eq:c_ex_2} with $-1$ and add it to Eq. \eqref{eq:c_ex_3}. Thus,
\begin{align}
    0 &< a_{51} - a_{11} < 1 \\
    1 &< a_{41} - a_{21} < 1.5
\end{align}
which is inconsistent with Eq. \eqref{eq:c_ex_nec}.

The following is a construction for the function in Eq. \eqref{eq:counterexample} using our techniques. We use $2$ significant digits to fit the matrix here.
\begin{equation}
\label{eq:new_B_construction}
\footnotesize\left[\begin{tabu}{cccccccc}
    \rowfont{\color{blue}}
  17.78 &  2.28 & -5.72 & -21.22 & -1.53 & -1.53 & -1.53 & -1.53\\
  \rowfont{\color{red}}
  19.28 &  3.28 & -4.72 & -20.22 & -0.53 & -0.53 & -0.53 & -0.53\\
  \rowfont{\color{blue}}
  20.28 &  4.78 & -3.72 & -19.22 &  0.47 &  0.47 &  0.47 &  0.47\\
  \rowfont{\color{red}}
  21.28 &  5.78 & -2.22 & -18.22 &  1.47 &  1.47 &  1.47 &  1.47\\
  \rowfont{\color{blue}}
  22.28 &  6.78 & -1.22 & -16.72 &  2.47 &  2.47 &  2.47 &  2.47
\end{tabu}\right]
\end{equation}

\subsection{Examples of NN Representations of Symmetric Boolean Functions without an Interval-Anchor Assignment}
\label{apx:ex_no_interval_anchor}
\begin{lemma}
    Let $f$ be an $n$-input symmetric Boolean function with $I(f) \leq 4$. Then, $NN(f) \geq I(f)$.
\end{lemma}
\begin{proof}
The cases when $I(f) \in \{1,2\}$ is trivial. When $I(f) = 3$, we can use Proposition \ref{prop:nec}. Now, assume that $3$-anchor NN representations exist for symmetric Boolean functions with $I(f) = 4$. Then, the assignment of anchors to the intervals can only belong to two cases: $a_1$ is assigned to the first and third intervals where $a_2$ \& $a_3$ shares the binary vectors for the second and fourth intervals, and vice versa. In both cases, there is a structure similar to \textbf{Case 1} of the proof of Lemma \ref{lem:3_interval}. By following similar steps, we can obtain a contradiction to Proposition \ref{prop:convex} and conclude that $3$-anchor NN representations do not exist for symmetric Boolean functions with $I(f) = 4$.
\end{proof}

Here is a 5-input counterexample for the function described in Lemma \ref{lem:3_interval} where the middle interval is shared by two positive anchors.

\begin{equation}
    \label{eq:no_int_anchor}
    \begin{tabular}{c|c}
        $|X|$ & $f(X)$ \\
        \cline{1-2}
         0 & \textcolor{red}{\textbf{0}} \\
         1 & \textcolor{blue}{1} \\
         2 & \textcolor{blue}{1} \\
         3 & \textcolor{blue}{1} \\
         4 & \textcolor{blue}{1} \\
         5 & \textcolor{red}{\textbf{0}}  
    \end{tabular}\hphantom{aa} A' =
    \begin{bmatrix}
        \textcolor{blue}{0} & \textcolor{blue}{0.57}  & \textcolor{blue}{0.57}  & \textcolor{blue}{0.57}  & \textcolor{blue}{0.57}  \\
        \textcolor{red}{\textbf{0.5}} & \textcolor{red}{\textbf{0.5}} & \textcolor{red}{\textbf{0.5}} & \textcolor{red}{\textbf{0.5}} & \textcolor{red}{\textbf{0.5}} \\
        \textcolor{blue}{1} & \textcolor{blue}{0.43}  & \textcolor{blue}{0.43}  & \textcolor{blue}{0.43}  & \textcolor{blue}{0.43}
    \end{bmatrix}
\end{equation}

\subsection{The Relationship Between NN Representations and Linear Threshold Circuits}
\label{apx:lt}

Another important aspect of NN representations is their place in the circuit class hierarchy. Suppose that we are given an $m$-anchor NN representation of a Boolean function. One can construct a linear threshold circuit of depth-3 with $O(m^2)$ many linear threshold gates \cite{hajnal2022nearest,hansen2015polynomial}. The converse, however, is not known; we cannot obtain an NN representation easily given a logic circuit.

We focus on the $\text{COMPARISON}$ (denoted by $\text{COMP}$) function to illustrate the size-resolution trade-off. $\text{COMP}(X,Y)$ is a linear threshold function $\text{COMP}(X,Y) = \mathds{1}\Big\{\sum_{i=1}^n 2^{i-1}x_i \geq \sum_{i=1}^n 2^{i-1}y_i\Big\}$ which computes whether an unsigned integer is greater than or equal to another. Here, $X$ and $Y$ corresponds to the binary expansions of these integers.

Let $\text{OR}$ denote Boolean OR function, $\text{AND}$ denote Boolean AND function, and $\text{THR}$ denote a linear threshold function. 
\begin{lemma}
\label{lem:circ_transformation}
Suppose that $f(X)$ is an $n$-input Boolean function with an $m$ anchor NN representation with the set of positive (or negative) anchors $P$ (or $N$). Then, there is a logic circuit that can compute $f(X)$ in the form $\textup{OR} \circ \textup{AND} \circ \textup{THR}$ of size $|P||N| + |P| + 1$ where $\circ$ denotes function composition.
\end{lemma}
\begin{proof}
    The label of the nearest neighbor can be found by the following formula.
    \begin{align}
        \label{eq:dist_formula}
        \arg\min_{i} d(a_i,X)^2 &= \arg\min_{i} |X| - 2(AX)_i + ||a_i||_2^2 \\
        &= \arg\max_{i} 2(AX)_i - ||a_i||_2^2
    \end{align}
    where $||.||_2$ denotes the Euclidean norm. Let $P$ (or $N$) be the set of positive (or negative) anchors $\{a_1,\dots,a_{|P|}\}$ (or, $\{b_1,\dots,b_{|N|}\}$.
    
    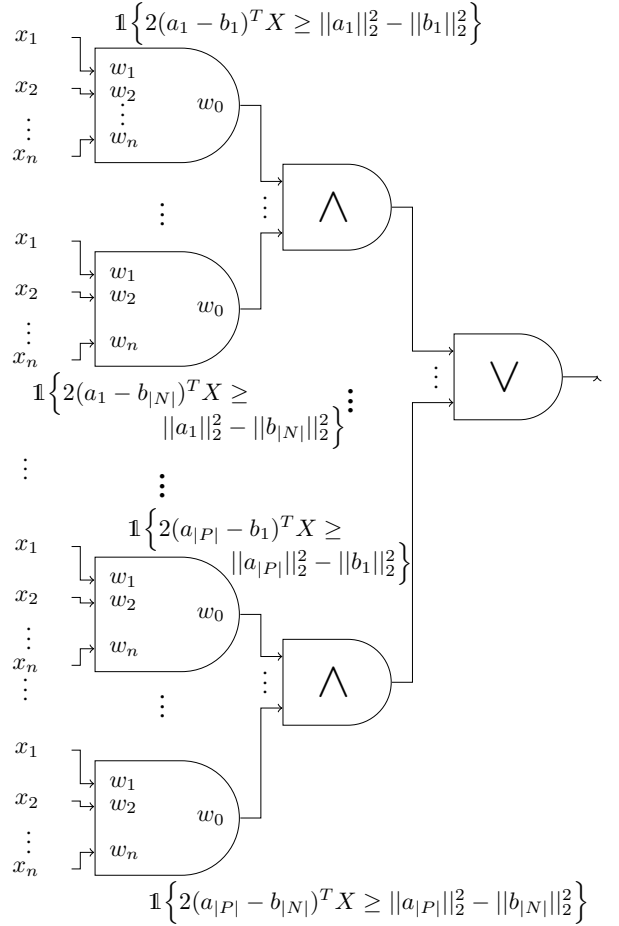
\begin{figure}[h]
    \centering
    \scalebox{0.9}{\begin{tikzpicture}
    \tikzstyle{sum} = [gate=white,label=center:+]
    \tikzstyle{input} = [circle]
    \newcommand{\inputnum}{9}
    \newcommand{\ltnum}{3}
    \newcommand{\andnum}{3}
    \newcommand{\eq}{=}
    
    \def\colors{red,violet,blue}
    \node[input,label=180:$x_1$] (i-1) at (0,0) {};
    \node[input,label=180:$x_2$] (i-2) at (0,-0.75) {};
    \node at (-0.45,-1.25) {$\vdots$};
    \node[input,label=180:$x_n$] (i-3) at (0,-1.75) {};
    \node[input,label=0:$\mathds{1}\Big\{2(a_1-b_1)^T X \geq ||a_1||_2^2  - ||b_1||_2^2 \Big\}$] (EQ-12-1) at (0.5,0.2) {};
    
    \node[and gate US, draw, logic gate inputs=nnnn, scale=2] (TH1) at (1.5,-1) {};
        \draw (i-1) -| +(0.3,0) |- (TH1.input 1);
        \draw (i-2) -| +(0.3,0) |- (TH1.input 2);
        \draw (i-3) -| +(0.3,0) |- (TH1.input 4);
        
    \node [above right=-0.25 and 0.25 of TH1.input 3] {$\vdots$};
    \node [right=0.1 of TH1.input 1] {$w_1$};
    \node [right=0.1 of TH1.input 2] {$w_2$};
    \node [right=0.1 of TH1.input 4] {$w_n$};
    \node [left =0.1 of TH1.output]{$w_0$};

    \node at (1.5,-2.5) {\large$\vdots$};
    
    \node[input,label=180:$x_1$] (i-4) at (0,-3) {};
    \node[input,label=180:$x_2$] (i-5) at (0,-3.75) {};
    \node at (-0.45,-4.25) {$\vdots$};
    \node[input,label=180:$x_n$] (i-6) at (0,-4.75) {};
    
    \node[input,label=0:$\mathds{1}\Big\{2(a_{1}-b_{|N|})^T X \geq$] (EQ-1n-1) at (-0.7,-5.25) {};
    \node[input,label=0:$||a_1||_2^2  - ||b_{|N|}||_2^2 \Big\}$] (EQ-1n-2) at (1.2,-5.75) {};
    
    \node[and gate US, draw, logic gate inputs=nnnn, scale=2] (TH2) at (1.5,-4) {};
        \draw (i-4) -| +(0.3,0) |- (TH2.input 1);
        \draw (i-5) -| +(0.3,0) |- (TH2.input 2);
        \draw (i-6) -| +(0.3,0) |- (TH2.input 4);
    \node [right=0.1 of TH2.input 1] {$w_1$};
    \node [right=0.1 of TH2.input 2] {$w_2$};
    \node [right=0.1 of TH2.input 4] {$w_n$};
    \node [left =0.1 of TH2.output ]{$w_0$};

    \node at (1.5,-6.5) {\LARGE$\vdots$};
    
    \node at (-0.5,-6.25) {$\vdots$};
    
    \node[input,label=180:$x_1$] (i-7) at (0,-7.5) {};
    \node[input,label=180:$x_2$] (i-8) at (0,-8.25) {};
    \node at (-0.45,-8.75) {$\vdots$};
    \node[input,label=180:$x_n$] (i-9) at (0,-9.25) {};
    \node[input,label=0:$\mathds{1}\Big\{2(a_{|P|}-b_{1})^T X \geq$] (EQ-P1-1) at (0.7,-7.25) {};
    \node[input,label=0:$||a_{|P|}||_2^2  - ||b_{1}||_2^2 \Big\}$] (EQ-P2-2) at (2.2,-7.75) {};
    \node[and gate US, draw, logic gate inputs=nnnn, scale=2] (TH3) at (1.5,-8.5) {};
        \draw (i-7) -| +(0.3,0) |- (TH3.input 1);
        \draw (i-8) -| +(0.3,0) |- (TH3.input 2);
        \draw (i-9) -| +(0.3,0) |- (TH3.input 4);
    \node [right=0.1 of TH3.input 1] {$w_1$};
    \node [right=0.1 of TH3.input 2] {$w_2$};
    \node [right=0.1 of TH3.input 4] {$w_n$};
    \node [left =0.1 of TH3.output ] {$w_0$};

    \node at (1.5,-9.75) {\large$\vdots$};
    
    \node at (-0.5,-9.5) {$\vdots$};
    \node[input,label=180:$x_1$] (i-7) at (0,-10.5) {};
    \node[input,label=180:$x_2$] (i-8) at (0,-11.25) {};
    \node at (-0.45,-11.75) {$\vdots$};
    \node[input,label=180:$x_n$] (i-9) at (0,-12.25) {};
    \node[input,label=0:$\mathds{1}\Big\{2(a_{|P|}-b_{|N|})^T X \geq ||a_{|P|}||_2^2  - ||b_{|N|}||_2^2 \Big\}$] (EQ-PN-1) at (1,-12.75) {};
    \node[and gate US, draw, logic gate inputs=nnnn, scale=2] (TH4) at (1.5,-11.5) {};
        \draw (i-7) -| +(0.3,0) |- (TH4.input 1);
        \draw (i-8) -| +(0.3,0) |- (TH4.input 2);
        \draw (i-9) -| +(0.3,0) |- (TH4.input 4);
    \node [right=0.1 of TH4.input 1] {$w_1$};
    \node [right=0.1 of TH4.input 2] {$w_2$};
    \node [right=0.1 of TH4.input 4] {$w_n$};
    \node [left =0.1 of TH4.output ] {$w_0$};
    \newcommand{\offset}{2.5}
    \node[and gate US, draw, logic gate inputs=nnnn, scale=1.5,label=center:\huge$\land$] (AND1) at (\offset+1.5,-2.5) {};
        \draw (TH1.output) -| +(0.3,0) |- (AND1.input 1);
        \node at (\offset+0.5,-2.4) {$\vdots$};
        \draw (TH2.output) -| +(0.3,0) |- (AND1.input 4);
        
    \node at (4.25,-5.25) {\LARGE$\vdots$};
    
    \node[and gate US, draw, logic gate inputs=nnnn, scale=1.5,label=center:\huge$\land$] (AND2) at (\offset+1.5,-9.5) {};
        \draw (TH3.output) -| +(0.3,0) |- (AND2.input 1);
        \node at (\offset+0.5,-9.4) {$\vdots$};
        \draw (TH4.output) -| +(0.3,0) |- (AND2.input 4);
    \renewcommand{\offset}{5}
    \node[and gate US, draw, logic gate inputs=nnnn, scale=1.5,label=center:\huge$\lor$] (OR) at (\offset+1.5,-5) {};
        \draw (AND1.output) -| +(0.3,0) |- (OR.input 1);
        \node at (\offset+0.5,-4.9) {$\vdots$};
        \draw (AND2.output) -| +(0.3,0) |- (OR.input 4);
        \draw (OR.output) -| ([xshift=0.5cm]OR.output);
        
    \end{tikzpicture}}
    \caption{A sketch of a transformation from an NN representation to a linear threshold circuit of depth 3. The first layer consists of linear threshold gates while the second and third layer is an $\text{AND}-\text{OR}$ network. The gates with $w_i$ inside with $w_0$ as the bias is $\text{THR}(X)$, the gates with $\land$ inside is $\text{AND}(X)$, and the gates with $\lor$ inside is $\text{OR}(X)$}
    \label{fig:lt_circuit_transformation}
    \end{figure}
    
    The constructive transformation is given in Fig. \ref{fig:lt_circuit_transformation}. In the first layer, we simply compare the distances of all individual positive anchors, say $a_i$, to all negative anchors $\{b_1,\dots,b_{|N|}\}$ using Eq. \eqref{eq:dist_formula}.  Hence, there are $|P||N|$ many threshold gates in the first layer. There is an AND gate corresponding to each positive anchor and the top gate is an OR gate, implying that the claimed circuit size is correct.
    
    Assume first that $f(X) = 1$. The linear threshold function
    \begin{equation}
        \label{eq:nn_to_lt}
        \mathds{1}\Big\{2a_{i}^T X - ||a_i||_2^2 \geq 2b_{j}^T X - ||b_j||_2^2 \Big\} 
    \end{equation}
    evaluates $1$ for some $i \in \{1,\dots,|P|\}$ and for all $j \in \{1,\dots,|N|\}$ because a positive anchor must be closer to $X$ than any negative anchor. Hence, the output of the corresponding $i^{th}$ AND gate will be $1$ and consequently, the output of the circuit is $1$.

    Conversely, assume $f(X) = 0$. Then, there always exist some $j$ such that Eq. \eqref{eq:nn_to_lt} is $0$ for all $i$ because the nearest neighbor to $X$ is negative (i.e. in $N$). Therefore, all of the second layer AND gates compute 0 so that the output of circuit is 0.
\end{proof}

To compute $\text{COMP}$, there is a lower bound on the number of threshold gates independent of the circuit depth given that there is a weight size constraint. Based on the lower bound and Lemma \ref{lem:circ_transformation}, we can obtain a similar bound on the NN complexity.

\begin{theorem}[\cite{kilic2021neural,roychowdhury1994lower}]
\label{th:comp_lower_bound}
Suppose that $X = (x_1,\dots,x_n)$ and $Y = (y_1,\dots,y_n)$ are binary vectors. The $\text{COMP}(X,Y)$ function can be computed by a linear threshold circuit of size $\Omega(n/\log{nW})$ where $W \in \mathbb{Z}$ is the maximum weight size.
\end{theorem}

\begin{lemma}
The number of anchors for an NN representation of an $n$-input $\textup{COMP}(X,Y)$ is $\Omega\Big(\sqrt{\frac{n}{\log{n}+r}}\Big)$ where $r$ is the resolution of the representation.
\end{lemma}
\begin{proof}
    Suppose that there is a NN representation of the $\textup{COMP}(X,Y)$ such that the number of anchors, say $m$, is $o\Big(\sqrt{\frac{n}{\log{n}+r}}\Big)$. Then, by Lemma \ref{lem:circ_transformation}, there is a depth-3 threshold circuit of size $o\Big(\frac{n}{\log{n}+r}\Big)$ computing $\text{COMP}(X,Y)$. Since $r = \lceil\log_2{W+1}\rceil$, this contradicts the lower bound for $\textup{COMP}(X,Y)$ given in Theorem \ref{th:comp_lower_bound}.
\end{proof}

\end{document}